\documentclass[journal]{IEEEtran}
\usepackage{graphicx}

\usepackage{bm}
\usepackage{amsfonts}
\usepackage{amsthm}
\usepackage{amssymb}

\usepackage{color}
\usepackage{mathrsfs}
\usepackage{flushend}

\usepackage[outdir=./]{epstopdf}

\usepackage{cite}

\usepackage{booktabs}

\newtheorem{proposition}{Proposition}

\ifCLASSINFOpdf
\else
\fi
\usepackage{amsmath}
\usepackage{amsthm}

\usepackage{algorithm}
\usepackage{algorithmic}
\ifCLASSOPTIONcompsoc
 \usepackage[caption=false,font=normalsize,labelfont=sf,textfont=sf]{subfig}
\else
 \usepackage[caption=false,font=footnotesize]{subfig}
\fi

\usepackage{stfloats}
\fnbelowfloat
\begin{document}
%
\title{Interfering Channel Estimation in Radar-Cellular Coexistence: How Much Information Do We Need?}
%
%

\author{Fan Liu,~\IEEEmembership{Student~Member,~IEEE,}
        Adrian Garcia-Rodriguez,~\IEEEmembership{Member,~IEEE,}\\
        Christos Masouros,~\IEEEmembership{Senior~Member,~IEEE,}
        and~Giovanni Geraci,~\IEEEmembership{Member,~IEEE}
\thanks{This work has been submitted to the IEEE for possible publication. Copyright may be transferred without notice, after which this version may no longer be accessible.}
\thanks{F. Liu and C. Masouros are with the Department of Electronic and Electrical Engineering, University College London, London, WC1E 7JE, UK. (e-mail: fan.liu@ucl.ac.uk, chris.masouros@ieee.org).}
\thanks{A. Garcia-Rodriguez and G. Geraci are with Nokia Bell Labs, Dublin 15, Ireland. (e-mail: a.garciarodriguez.2013@ieee.org, dr.giovanni.geraci@gmail.com).}
}

\maketitle
\begin{abstract}
In this paper, we focus on the coexistence between a MIMO radar and cellular base stations. We study the interfering channel estimation, where the radar is operated in the ``search and track" mode, and the BS receives the interference from the radar. Unlike the conventional methods where the radar and the cellular systems fully cooperate with each other, in this work we consider that they are uncoordinated and the BS needs to acquire the interfering channel state information (ICSI) by exploiting the radar probing waveforms. For completeness, both the line-of-sight (LoS) and Non-LoS (NLoS) channels are considered in the coexistence scenario. By further assuming that the BS has limited a priori knowledge about the radar waveforms, we propose several hypothesis testing methods to identify the working mode of the radar, and then obtain the ICSI through a variety of channel estimation schemes. Based on the statistical theory, we analyze the theoretical performance of both the hypothesis testing and the channel estimation methods. Finally, simulation results verify the effectiveness of our theoretical analysis and demonstrate that the BS can effectively estimate the interfering channel even with limited information from the radar.
\end{abstract}


%

\section{Introduction}
\IEEEPARstart {R}{ecent} years have witnessed an explosive growth of wireless services and devices. As a consequence, the frequency spectrum has become one of the most valuable resources. Since 2015, mobile network operators in the UK have been required to pay a combined annual total of \pounds 80.3 million for the 900MHz and \pounds 119.3 million for the 1800MHz bands \cite{UK_spectrum}. Given the crowdedness within the sub-10GHz band, policy regulators and network providers are now seeking for the opportunity to reuse spectrum currently restricted to other applications. Indeed, the frequency bands occupied for radar are among the best candidates to be shared among various communication systems in the near future \cite{federal2010connecting,ears,ssparc,specees}.
\subsection{Existing Approaches}
Aiming for realizing the spectral coexistence of radar and communication,  existing contributions mainly focus on mitigating the mutual interference between the two systems by use of precoding/beamforming techniques \cite{6817773,sodagari2012projection,babaei2013nullspace,khawar2015target,mahal2017spectral}. Such efforts can be found in the pioneering work of \cite{sodagari2012projection}, in which the radar signals are precoded by a so-called null-space projector (NSP), and thus the interference generated to the communication systems is zero-forced. To achieve a favorable performance trade-off, the NSP method is further improved in \cite{khawar2015target,mahal2017spectral} via Singular Value Decomposition (SVD), where the interference level can be adjusted considering the singular values of the channel matrix.
\\\indent As a step further, more recent works have exploited convex optimization techniques for jointly designing transmit waveforms/precoders of radar and communication systems, such that certain performance metrics can be optimized \cite{li2016mimo,li2016optimum,li2017joint,zheng2018joint,8352726,8233171,7898445,8254100,8355705,8351931}. For instance, in \cite{li2017joint}, the receive signal-to-interference-plus-noise ratio (SINR) of the radar is maximized in the presence of both the clutters and the communication interference, while the capacity of the communication system is guaranteed. The inverse problem has been tackled in \cite{zheng2018joint}, where the communication rate has been maximized subject to the radar SINR constraint, as well as the power budgets for both systems. While the aforementioned works are well-designed via sophisticated techniques, it is in general difficult for them to be applied to current radar applications, given the fact that the governmental and military agencies are unwilling to make major changes in their radar deployments, which may impose huge costs on their financial budgets \cite{UK_radar_plannning}. Hence, a more practical approach is to develop transmission schemes at the communication side only, where the radar is agnostic to the interference or even the operation of the communication system. In this line, [16] considers the coexistence between a MIMO radar and a BS performing multi-user MIMO (MU-MIMO) downlink transmissions, in which the precoder of the BS is the only optimization variable. In \cite{8254100,8355705}, the BS precoder has been further developed by exploiting the constructive multi-user interference, which demonstrates orders-of-magnitude power-savings.
\\\indent It is worth highlighting that precoding based techniques require the knowledge of the interfering channel either at the radar or the communication BS. In fact, perfect/imperfect channel state information (CSI) assumptions are quite typical in the above works. To obtain such information, the radar and the BS are supposed to fully cooperate with each other and transmit training symbols, in line with conventional channel estimation methods. In \cite{mahal2017spectral}, the MIMO radar needs to estimate the channel based on the received pilot signals sent by the BS, which inevitably occupies extra computational and signaling resources. Other works such as \cite{li2017joint} require an all-in-one control center to be connected to both systems via a dedicated side information link, which conducts the information exchange and the waveform optimization. In practical scenarios, however, the control center brings forward considerable complexity in the system design, and is thus difficult to implement. Moreover, since it is the cellular operator who exploits the spectrum of the radar, it is the performance of the latter that should be primarily guaranteed, i.e., the radar resources should be allocated to target detection rather than obtaining the CSI. Unfortunately, many existing contributions failed to address this issue, and, to the best of our knowledge, the channel estimation approaches tailored for the radar-cellular coexistence scenarios remain widely unexplored. In light of the above drawbacks regarding the CSI acquisition, the natural question is, 1) \emph{is it possible to estimate the channel when there is limited cooperation between the radar and the communication systems?} And if so, 2) \emph{how much information do we need for the estimation?}

\subsection{The Contribution of Our Work}
This paper aims at answering the above issues, where we focus on interfering channel estimation between a MIMO radar and a MIMO BS:
\begin{enumerate}
 \item To cope with the first issue above, we hereby propose to exploit the radar probing waveforms for estimating the interfering channel. In this case the radar does not need to send training symbols or estimate the channel by itself, and thus the need for cooperation is fully eliminated. Following the classic MIMO radar literature \cite{4350230,6324717}, we assume that the radar has two working modes, i.e., searching and tracking. In the search mode, the radar transmits a spatially orthogonal waveform, which formulates an omni-directional beampattern for searching potential targets over the whole angular domain. In the track mode, the radar transmits directional waveforms to track the target located at the angle of interest, and thus to obtain a more accurate observation. In the meantime, the BS is trying to estimate the channel based on the periodically received radar interference, which is tied to the radar's duty cycle. As the searching and tracking waveforms are randomly transmitted, we propose to identify the operation mode of the radar by use of the hypothesis testing approach, and then estimate the channel at the BS.
 \item To answer the second question raised above, we further investigate different cases under both LoS and NLoS channels, where different levels of priori knowledge about the radar waveforms are assumed to be known at the BS, i.e., from full knowledge of searching and tracking waveforms by the BS, to knowledge of searching waveform only, to a fully agnostic BS to the radar waveforms. From a realistic perspective, the second and the third cases are more likely to appear in practice while the first case serves as a performance benchmark. Accordingly, the theoretical performance analysis of the proposed approaches are provided.
 \end{enumerate}

For the sake of clarity, we summarize below the contributions of this paper:
\begin{enumerate}
\item We consider the interfering channel estimation for the coexistence of radar and cellular, where the radar probing waveforms are exploited to obtain ICSI at the BS.
\item We propose hypothesis testing approaches for the BS to identify the operation mode of the radar, based on the limited priori information available at the BS.
\item We analyze the theoretical performance of the proposed detectors and estimators, whose effectiveness is further verified via numerical results.
\end{enumerate}

The remainder of this paper is arranged as follows. Section II introduces the system model, Section III and Section IV propose interfering channel estimation approaches for NLoS and LoS scenarios, respectively. Subsequently, Section V analyzes the theoretical performance of the proposed schemes, Section VI provides the corresponding numerical results, and finally Section VII concludes the paper.
\\\indent {\emph{Notations}}: Unless otherwise specified, matrices are denoted by bold uppercase letters (i.e., $\mathbf{X}$), vectors are represented by bold lowercase letters (i.e., $\mathbf{z}$), and scalars are denoted by normal font (i.e., $\rho$). $\operatorname{tr}\left(\cdot\right)$ and $\operatorname{vec}\left(\cdot\right)$ denote the trace and the vectorization operations. $\otimes$ denotes the Kronecker product. $\left\| \cdot\right\|$ and $\left\| \cdot\right\|_F$ denote the $l_2$ norm and the Frobenius norm. $\left(\cdot\right)^T$, $\left(\cdot\right)^H$, and $\left(\cdot\right)^*$ stand for transpose, Hermitian transpose and complex conjugate, respectively.
\begin{figure}[!tp]
\centering
\subfloat[]{\includegraphics[width=0.9\columnwidth]{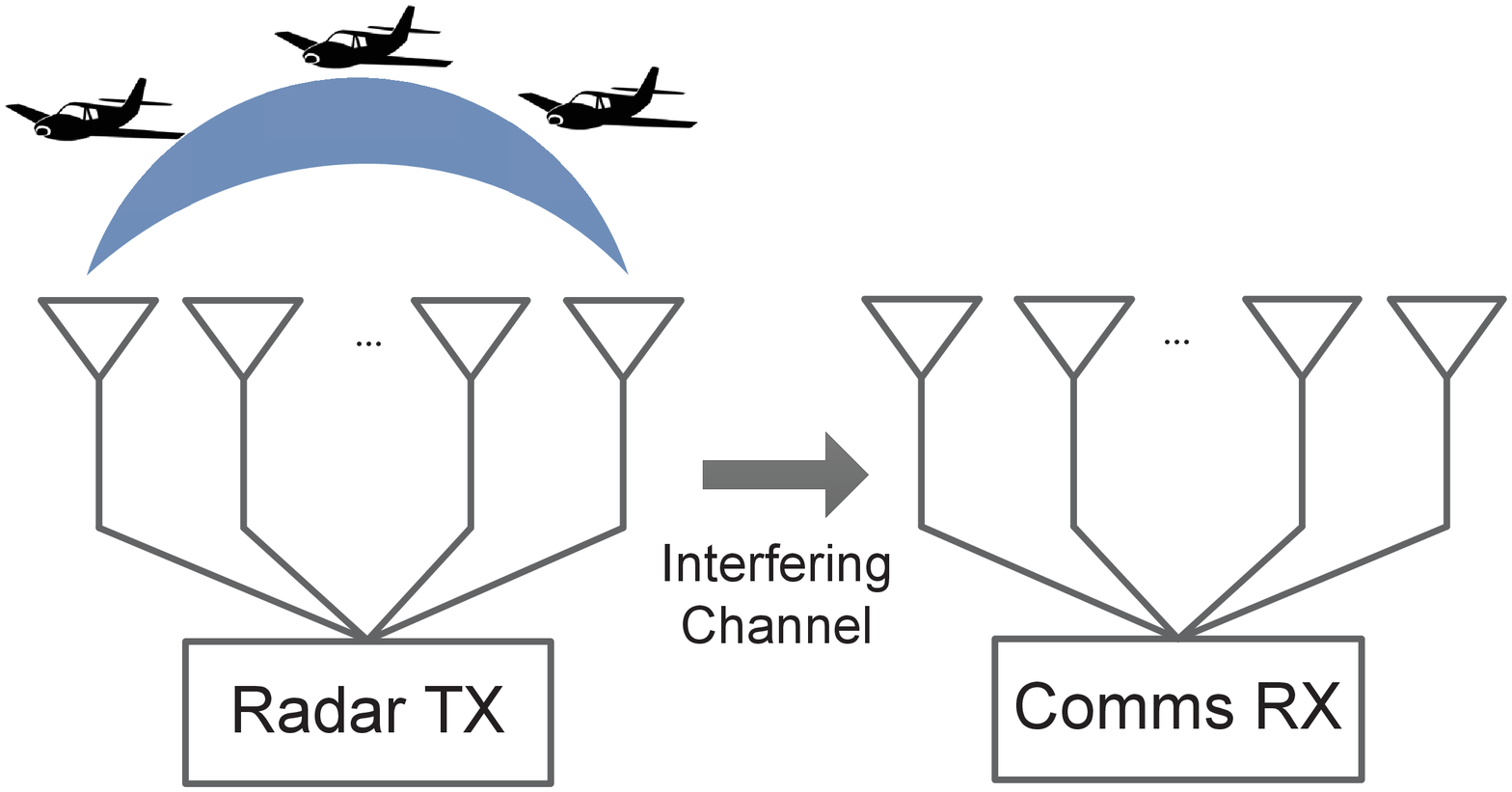}
\label{Radar_Search}}
\vspace{0.1in}
\hspace{.1in}
\subfloat[]{\includegraphics[width=0.9\columnwidth]{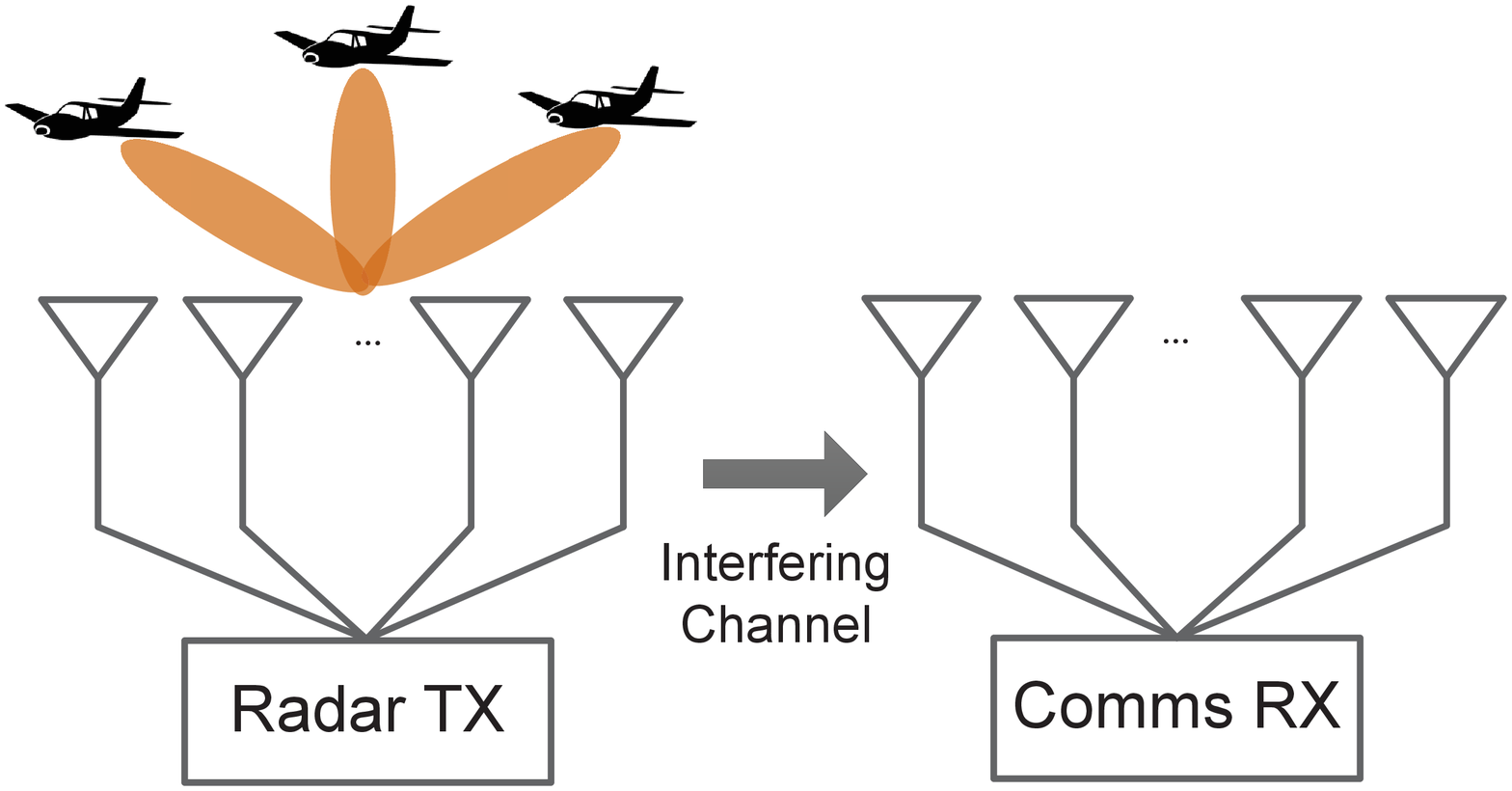}
\label{Radar_Track}}
\caption{MIMO Radar and BS coexistence. (a) Radar search mode; (b) Radar track mode.}
\label{fig_projector}
\end{figure}
%
%
%
%


\section{System Model}
As shown in Fig. 1, we consider a MIMO radar with $M_t$ transmit antennas and $M_r$ receive antennas that is detecting targets located in the far field. For simplicity, we assume that the MIMO radar employs the same antenna array for both transmission and reception, and denote ${M_t} = {M_r} = M$. Meanwhile, an \emph{N}-antenna BS operating in the same frequency band is receiving interference from the radar and trying to acquire the ICSI between them. Below we provide the system models for both the radar and the BS.
\subsection{Radar Signal Transmission - Search and Track}
It is widely known that by employing incoherent waveforms, the MIMO radar achieves higher Degrees of Freedom (DoFs) and better performance than the conventional phased-array radar that transmits correlated waveforms\cite{4350230}. By denoting the MIMO radar probing waveform as $\mathbf{X}\in \mathbb{C}^{M\times L}$, its spatial covariance matrix can be given as \cite{4350230,li2008mimo,6324717,4276989,4516997}
\begin{equation}\label{eq1}
    {{\mathbf{R}}_X} = \frac{1}{L}{\mathbf{X}}{{\mathbf{X}}^H},
\end{equation}
where $L$ is the length of the radar pulse. Throughout the paper we consider $L\ge N \ge M >2$, and assume uniform linear arrays (ULA) at both the radar and the BS. The corresponding beampattern can be thus given in the form \cite{4350230,li2008mimo,6324717,4276989,4516997}
\begin{equation}\label{eq2}
    {P_d}\left( \theta  \right) = {{\mathbf{a}}^H}\left( \theta  \right){{\mathbf{R}}_X}{\mathbf{a}}\left( \theta  \right),
\end{equation}
where $\theta$ denotes the azimuth angle,  and ${\mathbf{a}}\left( \theta \right) = \left[ {1,{e^{j2\pi \Delta \sin \left( \theta  \right)}},...,{e^{j2\pi \left( {M - 1} \right)\Delta \sin \left( \theta  \right)}}} \right]^T \in {\mathbb{C}^{{M} \times 1}}$ is the steering vector of the transmit antenna array with $\Delta$ being the spacing between adjacent antennas normalized by the wavelength.
\\\indent When the orthogonal waveform is transmitted by the MIMO radar, it follows that\cite{li2008mimo,7214226}
\begin{equation}\label{eq3}
    {{\mathbf{R}}_X} = \frac{{{P_R}}}{M}{{\mathbf{I}}_M},
\end{equation}
where $P_R$ is the total transmit power of the radar, and ${{\mathbf{I}}_M}$ is the \emph{M}-dimensional identity matrix.
\begin{figure}[!tp]
 \centering
 \includegraphics[width=0.9\columnwidth]{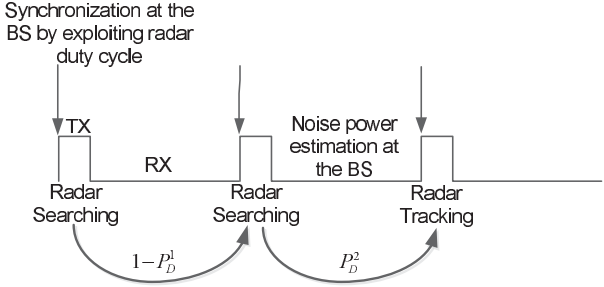}
 \caption{Radar working mode - ``search and track'' and operations performed by the communications BS.}
 \label{fig:FIG22}
\end{figure}
It is easy to see from (\ref{eq2}) that the covariance matrix (\ref{eq3}) generates an omni-directional beampattern, which is typically used for \emph{searching} when there is limited information about the target locations\cite{4350230}. Once a target is detected, the radar switches to the \emph{tracking} mode, where it will no longer transmit orthogonal waveforms and will generate a directional beampattern that points to the specific location, thus obtaining a more accurate observation. This, however, results in a non-orthogonal transmission, i.e., $ {{\mathbf{R}}_X} \ne \frac{{{P_R}}}{M}{{\mathbf{I}}_M}$. In this paper, we assume that the radar adopts both the searching and tracking modes subject to a probability transition model. This model is illustrated in Fig. 2 and can be summarized as follows \cite{4207172}:
\\\indent \emph{Assumption 1}: At the $i$-th pulse repetition interval (PRI) of the radar, the probability that the radar is operating at the tracking mode is $P_D^{\left(i-1\right)}$, where $P_D^{\left(i-1\right)}$ is the target detection probability of the $(i-1)$-th PRI.
\\\indent The above assumption entails that the MIMO radar changes its probing waveform $\mathbf{X}$ \emph{randomly} within each PRI, which makes it challenging for the BS to estimate the interfering channel between them.

\subsection{Interfering Channel Model}
The interfering channel between the BS and the radar could be characterized through different models, depending on their specific positions. For instance, the military and weather radars are typically located at high-altitude places such as top of the hills, in which case the channel between the BS and radar is likely to be a Line-of-Sight (LoS) channel. On the other hand, if the radar is used for monitoring the low-altitude flying objects (such as drones) or the urban traffic, it is usually deployed in urban areas at similar heights than the BS, thus resulting in a Non-Line-of-Sight (NLoS) channel. For completeness, we will discuss both cases in this paper. Since both the radar and the BS are located in fixed positions, we also adopt the following assumption:
\\\indent \emph{Assumption 2}: For the LoS coexistence scenario, we assume that the interfering channel from the radar to the BS is fixed. For the NLoS coexistence scenario, we assume the interfering channel is flat Rayleigh fading, and remains unchanged during several radar PRIs.

\subsection{BS Signal Reception Model}
Denoting the interfering channel as $\mathbf{G} \in \mathbb{C}^{N\times M}$, the received signal matrix at the BS can be given as
\begin{equation}\label{eq4}
    {\mathbf{Y}} = {\mathbf{G}}{{\mathbf{X}}} + {\mathbf{W}},
\end{equation}
where ${\mathbf{W}} = \left[ {{{\mathbf{w}}_1},{{\mathbf{w}}_2},...,{{\mathbf{w}}_L}} \right] \in {\mathbb{C}^{N \times L}}$ is the noise matrix, with ${{\mathbf{w}}_l} \sim \mathcal{C}\mathcal{N}\left( {{\mathbf{0}},{N_0}{{\mathbf{I}}_N}} \right),\forall l$. In the proposed hypothesis testing framework, the noise power $N_0$ plays an important role for normalizing the test statistic. Note that when radar keeps silent, the BS will receive nothing but the noise, and $N_0$ can be measured at this stage. Since the radar antenna number and its transmit power are fixed parameters, they can also be easily known to the BS operators. Therefore, it is reasonable to adopt the following assumption:
\\\indent \emph{Assumption 3}: The BS knows the value of $N_0$, $M$ and $P_R$.
\\\indent In order to estimate the channel and the noise power $N_0$, the BS needs to know when is radar transmitting, i.e., it must synchronize its clock with the radar pulses. As shown in Fig. 2, during one PRI, the radar only transmits for a portion of the time, typically below 10\%, and employs the remaining 90\% for receiving, during which the radar remains silent. Such a ratio is called \emph{duty cycle} \cite{zheng2018joint}. By exploiting this property, the BS is able to blindly estimate the beginning and the end of a radar pulse by some simple methods, such as energy detection. Note that for the NLoS channel scenario, there will be random time-spread delays within each pulse, which makes the synchronization inaccurate. However, since we assume a flat fading channel in the NLoS case, the time-spread delay will be contained within one snapshot of the radar, which results in negligible errors \cite{Tse2009Fundamentals}. We summarize the above through the following assumption:
\\\indent \emph{Assumption 4}: The BS can perfectly synchronize its clock with the radar pulses, i.e., it is able to know the beginning and the end of each radar pulse.

\subsection{Channel Estimation Procedure}
In light of the above discussion, we summarize below the channel estimation procedure at the BS:
\begin{enumerate}
\item Synchronize the system clock with the radar transmitted pulses.
\item Identify the working mode of the radar based on the received radar interference, i.e., whether the radar is searching or tracking.
\item Estimate the interfering channel by exploiting the limited knowledge about the radar waveforms.
\end{enumerate}
In the following, we will develop several approaches for the BS to acquire the ICSI when radar is randomly changing its probing waveform. We will first consider the NLoS channel case, and then the LoS channel case.

\section{NLoS Channel Scenario}
Consider the ideal case where the BS knows exactly the waveform sent by the radar in each PRI. Recalling (\ref{eq4}), the well-known maximum likelihood estimation (MLE) of the channel $\mathbf{G}$ is given as \cite{kay1998fundamentals}
\begin{equation}\label{eq5}
{\mathbf{\hat G}} = {\mathbf{YX}}^H{\left( {{{\mathbf{X}}}{\mathbf{X}}^H} \right)^{ - 1}},
\end{equation}
which is also known as the Least-Squares estimation (LSE) for ${\mathbf{G}}$. Unfortunately, the BS is not able to identify which waveform is transmitted, since the radar changes its waveform randomly at each PRI. Hence, (\ref{eq5}) can not be directly applied and it is difficult to estimate the channel directly. In this section, we discuss several cases where different levels of knowledge about the radar waveforms are available at the BS. At each level, we propose specifically tailored approaches to acquire the ICSI.
\subsection{BS Knows the Searching and Tracking Waveforms - Generalized Likelihood Ratio Test (GLRT)}
In this reference case, we assume that the BS knows both the searching and the tracking waveforms that the radar may transmit at the \emph{i}-th PRI, which we denote as ${\mathbf{X}}_0$ and ${\mathbf{X}}_1$, respectively. Since ${\mathbf{X}}_0$ is orthogonal, we have
\begin{equation}\label{eq6}
\frac{1}{L}{{\mathbf{X}}_0}{\mathbf{X}}_0^H = \frac{{{P_R}}}{M}{{\mathbf{I}}_M} \Rightarrow {{\mathbf{X}}_0}{\mathbf{X}}_0^H = \frac{{L{P_R}}}{M}{{\mathbf{I}}_M}.
\end{equation}
Before estimating the channel, the BS needs to identify which waveform is transmitted based on the received noisy data ${\mathbf{Y}}\in \mathbb{C}^{N \times L}$. This leads to the following hypothesis testing (HT) problem \cite{kay1998fundamentals_detection}
\begin{equation}\label{eq7}
{\mathbf{Y}} = \left\{ \begin{gathered}
  {\mathcal{H}_0}:{\mathbf{G}}{{\mathbf{X}}_0} + {\mathbf{W}}, \hfill \\
  {\mathcal{H}_1}:{\mathbf{G}}{{\mathbf{X}}_1} + {\mathbf{W}}. \hfill \\
\end{gathered}  \right.
\end{equation}
As per Assumption 1, the priori probabilities of the above two hypotheses can be given as
\begin{equation}\label{eq8}
P\left( {{\mathcal{H}_0}} \right) = 1 - P_D^{\left(i-1\right)},P\left( {{\mathcal{H}_1}} \right) = P_D^{\left(i-1\right)}.
\end{equation}
The HT problem (\ref{eq7}) can be solved via the generalized likelihood ratio test (GLRT), which is given by \cite{kay1998fundamentals_detection}
\begin{equation}\label{eq9}
\begin{gathered}
  {L_G}\left( {\mathbf{Y}} \right) = \frac{{p\left( {{\mathbf{Y}};{{{\mathbf{\hat G}}}_1},{\mathcal{H}_1}} \right)P\left( {{\mathcal{H}_1}} \right)}}{{p\left( {{\mathbf{Y}};{{{\mathbf{\hat G}}}_0},{\mathcal{H}_0}} \right)P\left( {{\mathcal{H}_0}} \right)}} \\
   = \frac{{p\left( {{\mathbf{Y}};{{{\mathbf{\hat G}}}_1},{\mathcal{H}_1}} \right)P_D^{\left(i-1\right)}}}{{p\left( {{\mathbf{Y}};{{{\mathbf{\hat G}}}_0},{\mathcal{H}_0}} \right)\left( {1 - P_D^{\left(i-1\right)}} \right)}}\mathop  \gtrless \limits_{{\mathcal{H}_0}}^{{\mathcal{H}_1}} \gamma, \\
\end{gathered}
\end{equation}
where $\gamma$ is the detection threshold, $p\left( {{\mathbf{Y}};{{{\mathbf{\hat G}}}},{\mathcal{H}_1}} \right)$ and $p\left( {{\mathbf{Y}};{{{\mathbf{\hat G}}}},{\mathcal{H}_0}} \right)$ are the likelihood functions (LFs), for the two hypotheses respectively, and can be given in the form
\begin{equation}\label{eq10}
\begin{gathered}
  p\left( {{\mathbf{Y}};{\mathbf{\hat G}},{\mathcal{H}_0}} \right) \hfill \\
   = {\left( {\pi {N_0}} \right)^{ - NL}}\exp \left( { - \frac{1}{{{N_0}}}\operatorname{tr} \left( {{{\left( {{\mathbf{Y}} - {\mathbf{\hat G}}{{\mathbf{X}}_0}} \right)}^H}\left( {{\mathbf{Y}} - {\mathbf{\hat G}}{{\mathbf{X}}_0}} \right)} \right)} \right), \hfill \\
\end{gathered}
\end{equation}
\begin{equation}\label{eq11}
\begin{gathered}
p\left( {{\mathbf{Y}};{\mathbf{\hat G}},{\mathcal{H}_1}} \right) \hfill \\
   = {\left( {\pi {N_0}} \right)^{ - NL}}\exp \left( { - \frac{1}{{{N_0}}}\operatorname{tr} \left( {{{\left( {{\mathbf{Y}} - {\mathbf{\hat G}}{{\mathbf{X}}_1}} \right)}^H}\left( {{\mathbf{Y}} - {\mathbf{\hat G}}{{\mathbf{X}}_1}} \right)} \right)} \right). \hfill \\
\end{gathered}
\end{equation}
In the above expressions, ${{{\mathbf{\hat G}}}_1}$ and ${{{\mathbf{\hat G}}}_0}$ are the MLEs under ${\mathcal{H}}_1$ and ${\mathcal{H}}_0$, which are obtained as
\begin{equation}\label{eq12}
{{\mathbf{\hat G}}_1} = {\mathbf{YX}}_1^H{\left( {{{\mathbf{X}}_1}{\mathbf{X}}_1^H} \right)^{ - 1}},
\end{equation}
\begin{equation}\label{eq13}
{{{\mathbf{\hat G}}}_0} = {\mathbf{YX}}_0^H{\left( {{{\mathbf{X}}_0}{\mathbf{X}}_0^H} \right)^{ - 1}} = \frac{M}{{L{P_R}}}{\mathbf{YX}}_0^H.
\end{equation}
Overall, once the BS determines which hypothesis to choose based on $\mathbf{Y}$, it can successfully estimate the channel by use of (\ref{eq12}) or (\ref{eq13}). However, it can be observed that the GLRT detector in (\ref{eq9}) requires the detection probability $P_D^{\left(i-1\right)}$ to be known to the BS, which is impossible in practice. Therefore, the detector (\ref{eq9}) can only serve as the \emph{optimal performance bound}. Since the actual $P_D^{\left(i-1\right)}$ is unknown to the BS, the reasonable choice for the priori probabilities is $P\left( {{\mathcal{H}_0}} \right) = P\left( {{\mathcal{H}_1}} \right) = 0.5$, namely $P_D^{\left(i-1\right)} = 0.5$. We can then apply the similar GLRT procedure for solving the HT problem. The test statistic in (\ref{eq9}) is thus simplified as
\begin{equation}\label{eq14}
{L_G}\left( {\mathbf{Y}} \right) = \frac{{p\left( {{\mathbf{Y}};{{{\mathbf{\hat G}}}_1},{\mathcal{H}_1}} \right)}}{{p\left( {{\mathbf{Y}};{{{\mathbf{\hat G}}}_0},{\mathcal{H}_0}} \right)}}\mathop  \gtrless \limits_{{\mathcal{H}_0}}^{{\mathcal{H}_1}} \gamma.
\end{equation}
\subsection{BS Knows Only the Searching Waveform - Rao Test}
In a realistic scenario, the tracking waveform ${\mathbf{X}}_1$ may vary from pulse to pulse. This is because the target to be detected may move very fast, which results in rapid changes in its parameters such as the distance, velocity and the azimuth angle. Hence, it is far from realistic to assume the BS knows ${\mathbf{X}}_1$, not to mention $P_D$ (in fact, both quantities are only determined after a target is detected). Nevertheless, as an omni-directional searching waveform, there is no reason for ${\mathbf{X}}_0$ to be changed rapidly. Indeed, in some cases, the radar may only use one waveform for omni-searching. Based on the above, to assume that the BS only knows ${\mathbf{X}}_0$ seems to be a more practical choice\footnote{At this stage we note the fact that such information exchange can be easily performed once prior to transmission, since the searching waveform of the radar remains unchanged. In contrast, conventional training based techniques require the radar or the BS to frequently send pilot symbols, which entails a much tighter cooperation between both systems.}. In this case, the HT problem (\ref{eq7}) can be recast as
\begin{equation}\label{eq15}
\begin{gathered}
  {\mathcal{H}_0}:{\mathbf{X}} = {{\mathbf{X}}_0},{\mathbf{G}}, \hfill \\
  {\mathcal{H}_1}:{\mathbf{X}} \ne {{\mathbf{X}}_0},{\mathbf{G}}. \hfill \\
\end{gathered}
\end{equation}
In (\ref{eq15}), the channel to be estimated is called the \emph{nuisance parameter} \cite{kay1998fundamentals_detection}.
\\\indent \emph{Remark 1:} At first glance, the GLRT procedure seems to be applicable to (\ref{eq15}) as well. However, note that to obtain the MLE of $\mathbf{G}$ under ${\mathcal{H}_1}$ is equivalent to solving the following optimization problem
\begin{equation}\label{eq16}
  \mathop {\min}\limits_{{\mathbf{G}},{\mathbf{X}}} \left\| {{\mathbf{Y}} - {\mathbf{GX}}} \right\|_F^2 \;\;s.t.\;\;\left\| {\mathbf{X}} \right\|_F^2 = L{P_R},\hfill \\
\end{equation}
where the constraint is to ensure the power budget of the radar-transmitted waveform. While the above problem is non-convex, it yields trivial solutions that achieve zero with a high probability. This is because the problem (\ref{eq16}) is likely to have more than enough DoFs to ensure that ${\mathbf{Y}} = {\mathbf{GX}}$, since $\mathbf{G}$ is unconstrained, and $\mathbf{X}$ can be always scaled to satisfy the norm constraint, where the scaling factor can be incorporated in $\mathbf{G}$. Therefore, the likelihood function under ${\mathcal{H}_1}$ will always be greater than that of ${\mathcal{H}_0}$, which makes the HT design meaningless.
\\\indent Realizing the fact above, we propose to use the Rao test (RT) to solve the HT problem (15), which does not need the MLE under ${\mathcal{H}_1}$. Based on \cite{7575639,liu2015rao,liu2014fisher}, let us define
\begin{equation}\label{eq17}
\begin{gathered}
  {\mathbf{\Theta }} = {\left[ {{{\operatorname{vec} }^T}\left( {\mathbf{X}} \right),{{\operatorname{vec} }^T}\left( {\mathbf{G}} \right)} \right]^T} \hfill \triangleq {\left[ {{\bm{\theta }}_r^T,{\bm{\theta }}_s^T} \right]^T}. \hfill \\
\end{gathered}
\end{equation}
Then, the RT statistic for the complex-valued parameters can be given in the form
\begin{equation}\label{eq18}
\begin{small}
\begin{gathered}
  {T_R}\left({{\mathbf{Y}}}\right)  \hfill \\
   = \left. {2\frac{{\partial \ln p\left( {{\mathbf{Y}};{\mathbf{\Theta }}} \right)}}{{\partial \operatorname{vec} \left( {\mathbf{X}} \right)}}} \right|_{{\mathbf{\Theta }} = {\mathbf{\tilde \Theta }}}^T{\left[ {{{\mathbf{J}}^{ - 1}}\left( {{\mathbf{\tilde \Theta }}} \right)} \right]_{{{\bm{\theta }}_r}{{\bm{\theta }}_r}}}{\left. {\frac{{\partial \ln p\left( {{\mathbf{Y}};{\mathbf{\Theta }}} \right)}}{{\partial {{\operatorname{vec} }^{\text{*}}}\left( {\mathbf{X}} \right)}}} \right|_{{\mathbf{\Theta }} = {\mathbf{\tilde \Theta }}}}\mathop  \gtrless \limits_{{H_0}}^{{H_1}} \gamma,  \hfill \\
\end{gathered}
\end{small}
\end{equation}
where ${\mathbf{\tilde \Theta }} = {\left[ {{\bm{\theta }}_r^T,{\bm{\hat \theta }}_s^T} \right]^T} = {\left[ {{{\operatorname{vec} }^T}\left( {{{\mathbf{X}}_0}} \right),{{\operatorname{vec} }^T}\left( {{{{\mathbf{\hat G}}}_0}} \right)} \right]^T}$ is the MLE under ${\mathcal{H}_0}$, and ${\left[ {{{\mathbf{J}}^{ - 1}}\left( {{\mathbf{\tilde \Theta }}} \right)} \right]_{{{\bm{\theta}} _r}{{\bm{\theta}} _r}}}$ is the upper-left partition of ${{{\mathbf{J}}^{ - 1}}\left( {{\mathbf{\tilde \Theta }}} \right)}$, with ${\mathbf{J}}\left( {\mathbf{\Theta }} \right)$ being the Fisher Information Matrix (FIM).
\\\indent Unlike the GLRT, the Rao test only lets the BS determine if the radar is using the searching mode, i.e., whether the orthogonal waveform matrix ${\mathbf{X}}_0$ is transmitted in the current radar PRI. In that case, the BS could obtain the MLE of the channel by use of (\ref{eq13}). Otherwise, the BS is required to wait until an orthogonal waveform is transmitted by the radar.
\subsection{Agnostic BS}
We now consider the hardest case that the BS does not know any of the waveforms transmitted by the radar. In this case, the BS still knows that ${{\mathbf{X}}}{\mathbf{X}}^H = \frac{{L{P_R}}}{M}{{\mathbf{I}}_M}$ for an omni-directional radar transmission. Therefore, the HT problem in (\ref{eq15}) can be recast as
\begin{equation}\label{eq19}
\begin{gathered}
  {\mathcal{H}_0}:{\mathbf{X}}{{\mathbf{X}}^H} = \frac{{L{P_R}}}{M}{{\mathbf{I}}_M},{\mathbf{G}}, \hfill \\
  {\mathcal{H}_1}:{\mathbf{X}}{{\mathbf{X}}^H} \ne \frac{{L{P_R}}}{M}{{\mathbf{I}}_M},{\mathbf{G}}. \hfill \\
\end{gathered}
\end{equation}
\\\indent \emph{Remark 2:} At first glance, we might be able to apply a generalized RT to solve the HT problem, where both the true values of ${\mathbf{G}}$ and ${\mathbf{X}}_0$ are replaced by their MLEs. This is because ${\mathbf{X}}_0$ is also unknown to the BS. Note that to obtain the MLEs of these two parameters is equivalent to solving the following optimization problem
\begin{equation}\label{eq20}
  \mathop {\min}\limits_{{\mathbf{G}},{\mathbf{X}}} \left\| {{\mathbf{Y}} - {\mathbf{GX}}} \right\|_F^2 \;\;s.t.\;\;{\mathbf{X}}{{\mathbf{X}}^H} = \frac{{L{P_R}}}{M}{{\mathbf{I}}_M}. \hfill \\
\end{equation}
Again, the above problem will unfortunately yield trivial solutions and make the HT design meaningless. This is because ${\mathbf{X}}$ can be viewed as a group of orthogonal basis, and the unconstrained ${\mathbf{G}}$ spans the whole space, which makes any given $\mathbf{Y}$ achievable with a high probability.
\\\indent The above remark involves that it is challenging to blindly estimate the ICSI for an agnostic BS under the NLoS channel scenario. Instead, we will show in the next section that blind channel estimation is feasible for the LoS channel scenario.

\section{LoS Channel Scenario}
In this section, we consider the scenario that the interfering channel between radar and BS is a LoS channel, where the received signal matrix at the BS is given by
\begin{equation}\label{eq21}
{\mathbf{Y}} = \alpha {\mathbf{b}}\left( \theta  \right){{\mathbf{a}}^H}\left( \theta  \right){\mathbf{X}} + {\mathbf{W}},
\end{equation}
where $\alpha$ represents the large-scale fading factor, $\theta$ is the angle of arrival (AoA) from the radar to the BS, ${\mathbf{b}}\left( \theta  \right) = {\left[ {1,{e^{j2\pi \Omega \sin \left( \theta  \right)}},...,{e^{j2\pi \left( {N - 1} \right)\Omega \sin \left( \theta  \right)}}} \right]^T} \in {\mathbb{C}^{N \times 1}}$ is the steering vector of the BS antenna array, with $\Omega$ being the normalized spacing, and ${\mathbf{a}}\left( \theta  \right)$ is radar's steering vector defined in Sec. II-A. Since the ULA geometry of the radar is fixed, we assume that the BS knows the spacing between the adjacent antennas of radar. Hence, the channel parameters that need to be estimated at the BS are $\alpha$ and $\theta$.
\\\indent Adopting the ideal assumption that the BS has instantaneous knowledge of the radar-transmitted waveform X in each PRI, the MLEs of the two parameters could be obtained by solving the optimization problem
\begin{equation}\label{eq22}
\mathop {\min }\limits_{\alpha ,\theta } \left\| {{\mathbf{Y}} - \alpha {\mathbf{b}}\left( \theta  \right){{\mathbf{a}}^H}\left( \theta  \right){\mathbf{X}}} \right\|_F^2.
\end{equation}
Note that if $\theta$ is fixed, the MLE of $\alpha$ can be given as
\begin{equation}\label{eq23}
\hat \alpha  = \frac{{{{\mathbf{b}}^H}\left( \theta  \right){\mathbf{Y}}{{\mathbf{X}}^H}{\mathbf{a}}\left( \theta  \right)}}{{L{{\left\| {{\mathbf{b}}\left( \theta  \right)} \right\|}^2}{{\mathbf{a}}^H}\left( \theta  \right){{\mathbf{R}}_X}{\mathbf{a}}\left( \theta  \right)}} = \frac{{{{\mathbf{b}}^H}\left( \theta  \right){\mathbf{Y}}{{\mathbf{X}}^H}{\mathbf{a}}\left( \theta  \right)}}{{NL{{\mathbf{a}}^H}\left( \theta  \right){{\mathbf{R}}_X}{\mathbf{a}}\left( \theta  \right)}},
\end{equation}
which suggests that the MLE of $\alpha$ depends on that of $\theta$. Substituting (\ref{eq23}) into the objective function of (\ref{eq22}), the MLE of $\theta$ can be thus given by
\begin{equation}\label{eq24}
\hat \theta  = \arg \mathop {\min }\limits_\theta  f\left( {{\mathbf{Y}};\theta ,{\mathbf{X}}} \right),
\end{equation}
where
\begin{equation}\label{eq25}
f\left( {{\mathbf{Y}};\theta ,{\mathbf{X}}} \right) = \left\| {{\mathbf{Y}} - \frac{{{{\mathbf{b}}^H}\left( \theta  \right){\mathbf{Y}}{{\mathbf{X}}^H}{\mathbf{a}}\left( \theta  \right){\mathbf{b}}\left( \theta  \right){{\mathbf{a}}^H}\left( \theta  \right){\mathbf{X}}}}{{NL{{\mathbf{a}}^H}\left( \theta  \right){{\mathbf{R}}_X}{\mathbf{a}}\left( \theta  \right)}}} \right\|_F^2.
\end{equation}
While (25) is non-convex, the optimum can be easily obtained through a 1-dimensional search over $\theta$.
\subsection{BS Knows the Searching and Tracking Waveforms - GLRT}
By assuming that the BS knows both $\mathbf{X}_0$ and $\mathbf{X}_1$, the HT problem (\ref{eq7}) can be reformulated as
\begin{equation}\label{eq26}
{\mathbf{Y}} = \left\{ \begin{gathered}
  {\mathcal{H}_0}:\alpha {\mathbf{b}}\left( \theta  \right){{\mathbf{a}}^H}\left( \theta  \right){{\mathbf{X}}_0} + {\mathbf{W}}, \hfill \\
  {\mathcal{H}_1}:\alpha {\mathbf{b}}\left( \theta  \right){{\mathbf{a}}^H}\left( \theta  \right){{\mathbf{X}}_1} + {\mathbf{W}}. \hfill \\
\end{gathered}  \right.
\end{equation}
The GLRT detector can be again applied to the LoS channel, in which case the likelihood functions under the two hypotheses are given as
\begin{equation}\label{eq27}
\begin{gathered}
  p\left( {{\mathbf{Y}};{{\hat \theta }_0},{\mathcal{H}_0}} \right) = {\left( {\pi {N_0}} \right)^{ - NL}}\exp \left( { - \frac{1}{{{N_0}}}f\left( {{\mathbf{Y}};{{\hat \theta }_0},{{\mathbf{X}}_0}} \right)} \right), \hfill \\
  p\left( {{\mathbf{Y}};{{\hat \theta }_1},{\mathcal{H}_1}} \right) = {\left( {\pi {N_0}} \right)^{ - NL}}\exp \left( { - \frac{1}{{{N_0}}}f\left( {{\mathbf{Y}};{{\hat \theta }_1},{{\mathbf{X}}_1}} \right)} \right), \hfill \\
\end{gathered}
\end{equation}
where $f$ is defined in (\ref{eq25}), and ${{\hat \theta }_0}$ and ${{\hat \theta }_1}$ are the MLEs of $\theta$ under the two hypotheses, respectively. By recalling (\ref{eq9}), the GLRT detector can be expressed as
\begin{equation}\label{eq28}
L_G^{LoS}\left( {\mathbf{Y}} \right) = \frac{1}{{{N_0}}}\left( {f\left( {{\mathbf{Y}};{{\hat \theta }_0},{{\mathbf{X}}_0}} \right) - f\left( {{\mathbf{Y}};{{\hat \theta }_1},{{\mathbf{X}}_1}} \right)} \right)\mathop  \gtrless \limits_{{\mathcal{H}_0}}^{{\mathcal{H}_1}} \gamma.
\end{equation}
The analytic distribution for (\ref{eq28}) is not obtainable, since there is no closed-form solution of $\hat \theta$ under both hypotheses.
\subsection{BS Knows Only the Searching Waveform - Energy Detection}
Similar to the NLoS channel case, a more practical assumption is to consider that the BS knows only the searching waveform $\mathbf{X}_0$. In this case, the GLRT detector is no longer applicable and the HT is given by
\begin{equation}\label{eq29}
\begin{gathered}
  {\mathcal{H}_0}:{\mathbf{X}} = {{\mathbf{X}}_0},\alpha ,\theta , \hfill \\
  {\mathcal{H}_1}:{\mathbf{X}} \ne {{\mathbf{X}}_0},\alpha ,\theta. \hfill \\
\end{gathered}
\end{equation}
At first glance, it seems that the Rao detector (\ref{eq18}) can be trivially extended from the NLoS channel scenario to the LoS case. Nevertheless, the following proposition puts an end to such a possibility.
\begin{proposition}
The Rao test does not exist for the scenario of the LoS channel.
\end{proposition}
\begin{proof}
See Appendix A.
\end{proof}
The algebraic explanation behind Proposition 1 is intuitive. As shown in (\ref{eq21}), by multiplying the rank-1 LoS channel to the radar waveform, the latter is mapped to a rank-1 subspace, which leads to serious information losses and yields a non-invertible FIM. Recalling (\ref{eq18}), the Rao test requires to compute the inverse of the FIM. Hence, it simply does not work in this specific case.
\\\indent To resolve the aforementioned issue, we consider an energy detection (ED) approach for the LoS channel. According to (\ref{eq21}), the average power of the received radar signal is given as
\begin{equation}\label{eq30}
\begin{gathered}
  {P_{LoS}} = \mathbb{E}\left( {\operatorname{tr} \left( {{\mathbf{Y}}{{\mathbf{Y}}^H}} \right)} \right) \hfill \\
   = \mathbb{E}\left( \begin{gathered}
  \operatorname{tr} \left( {{{\left| \alpha  \right|}^2}{\mathbf{b}}\left( \theta  \right){{\mathbf{a}}^H}\left( \theta  \right){\mathbf{X}}{{\mathbf{X}}^H}{\mathbf{a}}\left( \theta  \right){{\mathbf{b}}^H}\left( \theta  \right) + {\mathbf{W}}{{\mathbf{W}}^H}} \right) \hfill \\
  {\text{ +  }}2\operatorname{Re} \left( {\operatorname{tr} \left( {\alpha {\mathbf{b}}\left( \theta  \right){{\mathbf{a}}^H}\left( \theta  \right){\mathbf{X}}{{\mathbf{W}}^H}} \right)} \right) \hfill \\
\end{gathered}  \right) \hfill \\
   = \mathbb{E}\left( {\operatorname{tr} \left( {{{\left| \alpha  \right|}^2}{\mathbf{b}}\left( \theta  \right){{\mathbf{a}}^H}\left( \theta  \right){\mathbf{X}}{{\mathbf{X}}^H}{\mathbf{a}}\left( \theta  \right){{\mathbf{b}}^H}\left( \theta  \right){\text{ + }}{\mathbf{W}}{{\mathbf{W}}^H}} \right)} \right){\text{ }} \hfill \\
   \approx \frac{1}{L}\operatorname{tr} \left( {{{\left| \alpha  \right|}^2}{\mathbf{b}}\left( \theta  \right){{\mathbf{a}}^H}\left( \theta  \right){\mathbf{X}}{{\mathbf{X}}^H}{\mathbf{a}}\left( \theta  \right){{\mathbf{b}}^H}\left( \theta  \right)} \right) + N{N_0} \hfill \\
   = {\left| \alpha  \right|^2}{P_d}\left( \theta  \right)\operatorname{tr} \left( {{\mathbf{b}}\left( \theta  \right){{\mathbf{b}}^H}\left( \theta  \right)} \right) + N{N_0} \hfill \\
   = N{\left| \alpha  \right|^2}{P_d}\left( \theta  \right) + N{N_0}, \hfill \\
\end{gathered}
\end{equation}
where ${P_d}\left( \theta  \right)$ is the radar transmit beampattern defined in (\ref{eq2}), and the approximation in the fourth line of (\ref{eq30}) is based on the Law of Large Numbers.
\begin{figure}[!tp]
 \centering
 \includegraphics[width=0.9\columnwidth]{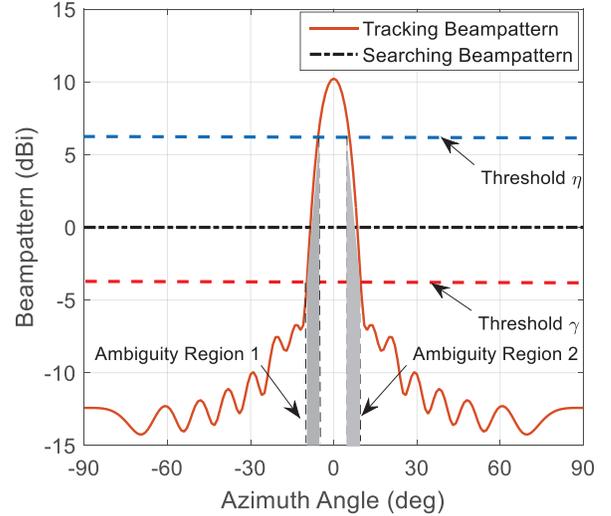}
 \caption{Searching and tracking beampatterns of the MIMO radar.}
 \label{fig:Search_and_Track_bmpts}
\end{figure}
From (\ref{eq30}), it is obvious that the received power at the BS is proportional to the radar's transmit power at the angle $\theta$. If the searching waveform $\mathbf{X}_0$ is transmitted, we have
\begin{equation}\label{eq31}
{P_d}\left( \theta  \right) = \frac{{{P_R}}}{M}{{\mathbf{a}}^H}\left( \theta  \right){{\mathbf{I}}_M}{\mathbf{a}}\left( \theta  \right) = {P_R},
\end{equation}
which suggests that the BS will receive equal power at each angle $\theta$. On the other hand, if the tracking waveform $\mathbf{X}_1$ is transmitted, most of the power will focus at the mainlobe, while less power will be distributed among the sidelobes, in which case the BS receives high power when it is located at the mainlobe of the radar, and much lower power at other angles. According to the aforementioned observations, in this paper we let the BS define two power measurement thresholds to determine whether the radar is in searching or tracking mode. As shown in Fig. 3\footnote{The tracking beampattern in Fig. 3 is generated based on the convex optimization method in \cite{4350230}, which we show in (\ref{eq_opt}) in Sec. VI.}, the BS chooses ${\mathcal{H}}_0$ if the received power falls between the two proposed thresholds, and it chooses ${\mathcal{H}}_1$ otherwise. Accordingly, the ED detector can be given as
\begin{equation}\label{eq32}
\begin{gathered}
  {T_E}\left({\mathbf{Y}}\right) = \frac{1}{L}\operatorname{tr} \left( {{\mathbf{Y}}{{\mathbf{Y}}^H}} \right) \in \left[ {\gamma ,\eta } \right] \to {\mathcal{H}_0}, \hfill \\
  {T_E}\left({\mathbf{Y}}\right) = \frac{1}{L}\operatorname{tr} \left( {{\mathbf{Y}}{{\mathbf{Y}}^H}} \right) \in \left( {0,\gamma } \right] \cup \left[ {\eta , + \infty } \right) \to {\mathcal{H}_1},
\end{gathered}
\end{equation}
where $\gamma$ and $\eta$ are the two power thresholds.
\\\indent \emph{Remark 3:} Note that the performance of the detector in (\ref{eq32}) depends on the size of the ambiguity regions shown in Fig. 3. By narrowing the distance between $\gamma$ and $\eta$, the detector trades-off the tolerance of the noise with the ambiguity area.
\\\indent By using the ED detector, the BS could choose from the two hypotheses without knowing both waveforms. Once $\mathcal{H}_0$ is chosen, the BS can estimate the AoA by finding the minimum of $f\left( {{\mathbf{Y}};\theta ,{\mathbf{X}}_0} \right)$.

\subsection{Agnostic BS}
Finally, we consider the hardest case where the BS does not know either the searching or tracking waveform. Note that the energy detector (\ref{eq32}) still works in this case, as it does not require any information about $\mathbf{X}_0$ or $\mathbf{X}_1$. The remaining question is how to estimate the channel. In order to do so, we first note that for the case of omni-directional transmission we have
\begin{equation}\label{eq33}
{P_d}\left( \theta  \right) = \frac{{{P_R}}}{M}{{\mathbf{a}}^H}\left( \theta  \right){{\mathbf{I}}_M}{\mathbf{a}}\left( \theta  \right) = {P_R},
\end{equation}
in which case (\ref{eq31}) can be rewritten as
\begin{equation}\label{eq34}
\begin{gathered}
  {P_{LoS}} = \mathbb{E}\left( {\operatorname{tr} \left( {{\mathbf{Y}}{{\mathbf{Y}}^H}} \right)} \right) \hfill \\
   \approx \frac{1}{L}\operatorname{tr} \left( {{\mathbf{Y}}{{\mathbf{Y}}^H}} \right) \approx N{P_R}{\left| \alpha  \right|^2} + N{N_0}. \hfill \\
\end{gathered}
\end{equation}
From (\ref{eq34}), it follows that
\begin{equation}\label{eq35}
{\left| \alpha  \right|^2} \approx \frac{{\operatorname{tr} \left( {{\mathbf{Y}}{{\mathbf{Y}}^H}} \right)}}{{LN{P_R}}} - \frac{{{N_0}}}{{{P_R}}},
\end{equation}
which can be used for estimating the absolute value of $\alpha$. It can be further observed that for the omni-directional transmission, we also have
\begin{equation}\label{eq36}
\begin{gathered}
  \frac{1}{L}{\mathbf{Y}}{{\mathbf{Y}}^H} = \frac{{{{\left| \alpha  \right|}^2}{P_R}}}{M}{\mathbf{b}}\left( \theta  \right){{\mathbf{a}}^H}\left( \theta  \right){{\mathbf{I}}_M}{\mathbf{a}}\left( \theta  \right){{\mathbf{b}}^H}\left( \theta  \right) + {\mathbf{\tilde W}} \hfill \\
   = {\left| \alpha  \right|^2}{P_R}{\mathbf{b}}\left( \theta  \right){{\mathbf{b}}^H}\left( \theta  \right) + {\mathbf{\tilde W}}, \hfill \\
\end{gathered}
\end{equation}
where ${\mathbf{\tilde W}}$ is the noise matrix. The LSE of $\theta$ can be thus given by
\begin{equation}\label{eq37}
\hat \theta  = \arg \mathop {\min }\limits_\theta  \left\| {\frac{{{\mathbf{Y}}{{\mathbf{Y}}^H}}}{{L{P_R}}} - {{\left| \alpha  \right|}^2}{\mathbf{b}}\left( \theta  \right){{\mathbf{b}}^H}\left( \theta  \right)} \right\|_F^2.
\end{equation}
\\\indent Overall, once $\mathcal{H}_0$ is chosen by the energy detector (\ref{eq32}), one can estimate ${{\left| \alpha  \right|}^2}$ and $\theta$ by (\ref{eq35}) and (\ref{eq37}) respectively, even without any knowledge about the radar waveforms. We remark that, since the noise matrix ${\mathbf{\tilde W}}$ is no longer Gaussian distributed, (\ref{eq35}) and (\ref{eq37}) are not MLEs of the parameters.
\\\indent For clarity, we summarize the proposed approaches for different scenarios in Table. I.
\begin{table}
\renewcommand{\arraystretch}{1.3}
\caption{Proposed Approaches for Different Scenarios}
\label{table_example}
\centering
\begin{tabular}{lcc}
\toprule
\bf  & \bf{NLoS Channel} & \bf{LoS Channel} \\
\midrule
BS Knows Both Waveforms & GLRT & GLRT  \\
BS Knows Searching Waveform & Rao Test & Energy Detection  \\
Agnostic BS & None & Energy Detection \\
\bottomrule
\end{tabular}
\end{table}


\section{Theoretical Performance Analysis}
In this section, we provide the theoretical performance analysis for the proposed hypothesis testing and channel estimation approaches. With this purpose, we use decision error probability and the mean squared error (MSE) as performance metrics.
\subsection{GLRT for NLoS Channels}
To analyze the performance of the GLRT detector, the MLEs of the unknown parameters under different hypotheses must be derived in closed-forms. While we consider GLRT for both NLoS and LoS channels in the previous discussion, the closed-form MLE of the AoA is not obtainable for the LoS channel. Therefore, we will only analyze the GLRT performance for the NLoS channel in this subsection. Firstly, let us substitute (\ref{eq12}) and (\ref{eq13}) into (\ref{eq10}) and (\ref{eq11}), which yield
\begin{equation}\label{eq38}
\begin{small}
\begin{gathered}
  p\left( {{\mathbf{Y}};{{{\mathbf{\hat G}}}_0},{\mathcal{H}_0}} \right) \hfill \\
   = {\left( {\pi {N_0}} \right)^{ - NL}}\exp \left( { - \frac{1}{{{N_0}}}\operatorname{tr} \left( {{{\left( {{\mathbf{Y}} - {{{\mathbf{\hat G}}}_0}{{\mathbf{X}}_0}} \right)}^H}\left( {{\mathbf{Y}} - {{{\mathbf{\hat G}}}_0}{{\mathbf{X}}_0}} \right)} \right)} \right) \hfill \\
   = {\left( {\pi {N_0}} \right)^{ - NL}}\exp \left( { - \frac{1}{{{N_0}}}\operatorname{tr} \left( {{\mathbf{Y}}\left( {{\mathbf{I}} - \frac{M}{{L{P_R}}}{\mathbf{X}}_0^H{{\mathbf{X}}_0}} \right){{\mathbf{Y}}^H}} \right)} \right), \hfill \\
\end{gathered}
\end{small}
\end{equation}
and
\begin{equation}\label{eq39}
\begin{small}
\begin{gathered}
  p\left( {{\mathbf{Y}};{{{\mathbf{\hat G}}}_1},{\mathcal{H}_1}} \right) \hfill \\
   = {\left( {\pi {N_0}} \right)^{ - NL}}\exp \left( { - \frac{1}{{{N_0}}}\operatorname{tr} \left( {{{\left( {{\mathbf{Y}} - {{{\mathbf{\hat G}}}_1}{{\mathbf{X}}_1}} \right)}^H}\left( {{\mathbf{Y}} - {{{\mathbf{\hat G}}}_1}{{\mathbf{X}}_1}} \right)} \right)} \right) \hfill \\
   = {\left( {\pi {N_0}} \right)^{ - NL}}\exp \left( { - \frac{1}{{{N_0}}}\operatorname{tr} \left( {{\mathbf{Y}}\left( {{\mathbf{I}} - {\mathbf{X}}_1^H{{\left( {{{\mathbf{X}}_1}{\mathbf{X}}_1^H} \right)}^{ - 1}}{{\mathbf{X}}_1}} \right){{\mathbf{Y}}^H}} \right)} \right). \hfill \\
\end{gathered}
\end{small}
\end{equation}
\begin{figure*}[ht]
\normalsize
\newcounter{MYtempeqncnt1}
\setcounter{MYtempeqncnt1}{\value{equation}}
\setcounter{equation}{50}
\begin{equation}\label{eq51}
{T_R}\left({{\mathbf{Y}}}\right)  = \frac{2}{{{N_0}}}\operatorname{tr} \left( {\left( {{{\mathbf{I}}_L} - \frac{M}{{L{P_R}}}{\mathbf{X}}_0^H{{\mathbf{X}}_0}} \right){{\mathbf{Y}}^H}{\mathbf{YX}}_0^H{{\left( {{{\mathbf{X}}_0}{{\mathbf{Y}}^H}{\mathbf{YX}}_0^H} \right)}^{ - 1}}{{\mathbf{X}}_0}{{\mathbf{Y}}^H}{\mathbf{Y}}} \right)\mathop  \gtrless \limits_{{\mathcal{H}_0}}^{{\mathcal{H}_1}} \gamma.
\end{equation}
\setcounter{equation}{\value{MYtempeqncnt1}}
\hrulefill
\vspace*{4pt}
\end{figure*}
Taking the logarithm of (\ref{eq9}) we obtain
\begin{equation}\label{eq40}
\begin{gathered}
  \ln \frac{{p\left( {{\mathbf{Y}};{{{\mathbf{\hat G}}}_1},{\mathcal{H}_1}} \right)P_D^{\left(i-1\right)}}}{{p\left( {{\mathbf{Y}};{{{\mathbf{\hat G}}}_0},{\mathcal{H}_0}} \right)\left( {1 - P_D^{\left(i-1\right)}} \right)}} \hfill \\
   = \frac{1}{{{N_0}}}\operatorname{tr} \left( {{\mathbf{Y}}\left( {{\mathbf{X}}_1^H{{\left( {{{\mathbf{X}}_1}{\mathbf{X}}_1^H} \right)}^{ - 1}}{{\mathbf{X}}_1} - \frac{M}{{L{P_R}}}{\mathbf{X}}_0^H{{\mathbf{X}}_0}} \right){{\mathbf{Y}}^H}} \right) \hfill \\
  \;\;\; - \ln \left( {\frac{{1 - {P_D^{\left(i-1\right)}}}}{{{P_D^{\left(i-1\right)}}}}} \right)\mathop  \gtrless \limits_{{\mathcal{H}_0}}^{{\mathcal{H}_1}} \gamma_0.  \hfill \\
\end{gathered}
\end{equation}
Finally, the GLRT detector can be given as
\begin{equation}\label{eq41}
\begin{small}
\begin{gathered}
  {L_G}\left({{\mathbf{Y}}}\right)  = \frac{1}{{{N_0}}}\operatorname{tr} \left( {{\mathbf{Y}}\left( {{\mathbf{X}}_1^H{{\left( {{{\mathbf{X}}_1}{\mathbf{X}}_1^H} \right)}^{ - 1}}{{\mathbf{X}}_1} - \frac{M}{{L{P_R}}}{\mathbf{X}}_0^H{{\mathbf{X}}_0}} \right){{\mathbf{Y}}^H}} \right) \hfill \\
  \;\;\;\;\mathop  \gtrless \limits_{{\mathcal{H}_0}}^{{\mathcal{H}_1}}\gamma = \gamma_0 + \ln \left( {\frac{{1 - P_D^{\left(i-1\right)}}}{{P_D^{\left(i-1\right)}}}} \right). \hfill \\
\end{gathered}
\end{small}
\end{equation}
Note that both ${{\mathbf{X}}_1^H{{\left( {{{\mathbf{X}}_1}{\mathbf{X}}_1^H} \right)}^{ - 1}}{{\mathbf{X}}_1}}$ and ${\frac{M}{{L{P_R}}}{\mathbf{X}}_0^H{{\mathbf{X}}_0}}$ are projection matrices \cite{kay1998fundamentals}. The physical meaning of (\ref{eq41}) is intuitive, i.e., to project the received signal onto the row spaces of ${\mathbf{X}}_1$ and ${\mathbf{X}}_0$ respectively, and to compute the difference between the lengths of the projections to decide which hypothesis to choose. Letting $P_D^{\left(i-1\right)}=0.5$, we have $\ln \left( {\frac{{1 - P_D^{\left(i-1\right)}}}{{P_D^{\left(i-1\right)}}}} \right)=0$ and $\gamma = \gamma_0$, which represents the case that $P_D$ is unknown.
\\\indent We now derive the Cumulative Distribution Function (CDF) of $L_G$. Defining
\begin{equation}\label{eq42}
\begin{gathered}
  {\mathbf{A}} = {\mathbf{X}}_1^H{\left( {{{\mathbf{X}}_1}{\mathbf{X}}_1^H} \right)^{ - 1}}{{\mathbf{X}}_1},{\mathbf{B}} = \frac{M}{{L{P_R}}}{\mathbf{X}}_0^H{{\mathbf{X}}_0}, \hfill \\
  {\mathbf{\tilde y}} = \frac{{\operatorname{vec} \left( {{{\mathbf{Y}}^H}} \right)}}{{\sqrt {{N_0}} }},{\mathbf{D}} = {{\mathbf{I}}_N} \otimes \left( {{\mathbf{A}} - {\mathbf{B}}} \right), \hfill \\
\end{gathered}
\end{equation}
it follows that
\begin{equation}\label{eq43}
{L_G}\left({{\mathbf{Y}}}\right) = {{{\mathbf{\tilde y}}}^H}\left( {{{\mathbf{I}}_N} \otimes \left( {{\mathbf{A}} - {\mathbf{B}}} \right)} \right){\mathbf{\tilde y}} = {{{\mathbf{\tilde y}}}^H}{\mathbf{D\tilde y}}.
\end{equation}
If $\mathbf{D}$ is an idempotent matrix, then the test statistic subjects to the non-central chi-squared distribution \cite{kay1998fundamentals}. While both $\mathbf{A}$ and $\mathbf{B}$ are idempotent, it is not clear if their difference is still idempotent. Moreover, their is no guarantee that $\mathbf{D}$ is semidefinite. Hence, $\mathbf{D}$ is an indefinite matrix in general, which makes ${L_G}$ an indefinite quadratic form (IQF) in Gaussian variables.
\\\indent Given the non-zero mean value of ${\mathbf{\tilde y}}$, ${L_G}$ becomes a non-central Gaussian IQF, which is known to have no closed-form expression for its CDF \cite{4712719,mathai1992quadratic}. Based on \cite{7312953}, here we consider a so-called saddle-point method to approximate the CDF of the test statistic. It is clear that ${\mathbf{\tilde y}} \sim \mathcal{C}\mathcal{N}\left( {{\mathbf{b}},{{\mathbf{I}}_{NL}}} \right)$, where
\begin{equation}\label{eq44}
{\mathbf{b}} = \left\{ \begin{gathered}
  {{{\mathcal{H}_0}:\operatorname{vec} \left( {{\mathbf{X}}_0^H{{\mathbf{G}}^H}} \right)} \mathord{\left/
 {\vphantom {{{\mathcal{H}_0}:\operatorname{vec} \left( {{\mathbf{X}}_0^H{{\mathbf{G}}^H}} \right)} {\sqrt {{N_0}} ,}}} \right.
 \kern-\nulldelimiterspace} {\sqrt {{N_0}} ,}} \hfill \\
  {{{\mathcal{H}_1}:\operatorname{vec} \left( {{\mathbf{X}}_1^H{{\mathbf{G}}^H}} \right)} \mathord{\left/
 {\vphantom {{{\mathcal{H}_1}:\operatorname{vec} \left( {{\mathbf{X}}_1^H{{\mathbf{G}}^H}} \right)} {\sqrt {{N_0}} ,}}} \right.
 \kern-\nulldelimiterspace} {\sqrt {{N_0}} ,}} \hfill \\
\end{gathered}  \right.
\end{equation}
which are the mean values for ${{\mathbf{\tilde y}}}$ under $\mathcal{H}_0$ and $\mathcal{H}_1$ respectively. Let us denote the eigenvalue decomposition of $\mathbf{D}$ as ${\mathbf{D}} = {\mathbf{Q\Lambda }}{{\mathbf{Q}}^H}$, where ${\mathbf{\Lambda }} = \operatorname{diag} \left( {{\lambda _1},{\lambda _2},...,{\lambda _{NL}}} \right)$ contains the eigenvalues. Based on the saddle-point approximation \cite{7312953}, the CDF of ${L_G}$ is given as
\begin{equation}\label{eq45}
P\left( {{L_G} \le \gamma } \right) \approx \frac{1}{{2\pi }}\exp \left( {s\left( {{\omega _0}} \right)} \right)\sqrt {\frac{{2\pi }}{{\left| {s''\left( {{\omega _0}} \right)} \right|}}},
\end{equation}
where
\begin{equation}\label{eq46}
s\left( \omega  \right) = \ln \left( {\frac{{{e^{\gamma \left( {j\omega  + \beta } \right)}}{e^{ - c\left( \omega  \right)}}}}{{\left( {j\omega  + \beta } \right)\det \left( {{\mathbf{I}} + \left( {j\omega  + \beta } \right){\mathbf{\Lambda }}} \right)}}} \right),
\end{equation}
\begin{equation}\label{eq47}
c\left( \omega  \right) = \sum\limits_{i = 1}^{NL} {{{\left| {{{\bar b}_i}} \right|}^2}}  - \sum\limits_{i = 1}^{NL} {\frac{{{{\left| {{{\bar b}_i}} \right|}^2}}}{{1 - \left( {j\omega  + \beta } \right){\lambda _i}}}},
\end{equation}
\begin{equation}\label{eq48}
{\mathbf{\bar b}} = {{\mathbf{Q}}^H}{\mathbf{b}} = {\left[ {{{\bar b}_1},{{\bar b}_2},...,{{\bar b}_{NL}}} \right]^T}.
\end{equation}
The above results hold for any $\beta > 0$. $\omega_0$ is the so-called saddle point, which is the solution of the following equation
\begin{equation}\label{eq49}
\begin{gathered}
  s'\left( {j\omega } \right) =  - \frac{1}{{\left( { - \omega  + \beta } \right)}} - \sum\limits_{i = 1}^{NL} {\frac{{{\lambda _i}}}{{1 + {\lambda _i}\left( { - \omega  + \beta } \right)}}}  \hfill \\
  \;\;\;\;\;\;\;\;\;\;\;\;\;\;\; + \gamma  - \sum\limits_{i = 1}^{NL} {\frac{{{{\left| {{{\bar b}_i}} \right|}^2}{\lambda _i}}}{{{{\left( {1 + {\lambda _i}\left( { - \omega  + \beta } \right)} \right)}^2}}}}  = 0, \hfill \\
\end{gathered}
\end{equation}
where $\omega  = j\left( {\beta  + p} \right)$. It has been proved that (\ref{eq49}) has a single real solution on $p \in \left( { - \infty ,0} \right)$ \cite{7312953}, which can be numerically found through a 1-dimensional searching.
\\\indent At the \emph{i}-th PRI, it is natural to measure the performance of GLRT by use of the decision error probability given the CDF of $L_G$, which is obtained as
\begin{equation}\label{eq50}
\begin{gathered}
  P_G^{\left(i\right)} = P\left( {{L_G} \ge \gamma ;{\mathcal{H}_0}} \right)P\left( {{\mathcal{H}_0}} \right) + P\left( {{L_G} \le \gamma ;{\mathcal{H}_1}} \right)P\left( {{\mathcal{H}_1}} \right) \hfill \\
  = \left( {1 - P\left( {{L_G} \le \gamma ;{\mathcal{H}_0}} \right)} \right)\left( {1 - P_D^{\left(i-1\right)}} \right) \hfill \\
  \;\;\;\;\;\;\;\;\;\;\;\;\;\;\;\;\;\;\;\;\;\;\;\;\;\;\;\;\;\;\;\;\;\;\;\; + P\left( {{L_G} \le \gamma ;{\mathcal{H}_1}} \right)P_D^{\left(i-1\right)}, \hfill \\
\end{gathered}
\end{equation}
where the CDF of ${L_G}$ under each hypothesis can be computed using the above equations (\ref{eq45})-(\ref{eq49}), by accordingly substituting the values of ${\mathbf{b}}$ under the two hypotheses, which are given in (\ref{eq44}).
\subsection{Rao Test for NLoS Channels}
We start from the following proposition.
\begin{proposition}
The Rao detector for solving (\ref{eq15}) is given by (\ref{eq51}), shown at the top of this page.
\end{proposition}
\begin{proof}
See Appendix B.
\end{proof}
It is clear from (\ref{eq51}) that we do not need any information about $\mathbf{X}_1$ for solving the HT problem (\ref{eq15}), which makes it a suitable detector for the practical scenario where the BS only knows $\mathbf{X}_0$. While $\mathbf{Y}$ is Gaussian distributed, it is very difficult to analytically derive the CDF of (\ref{eq51}) due to the highly non-linear operations involved. By realizing this, here we only focus our attention on a special case, where the distribution becomes tractable. Note that if $L\ge M = N$ holds true, ${\mathbf{YX}}_0^H \in \mathbb{C}^{N \times N}$ and ${{\mathbf{X}}_0}{{\mathbf{Y}}^H}\in \mathbb{C}^{N \times N}$ become the invertible square matrices with a high probability, in which case we have
\setcounter{equation}{51}
\begin{equation}\label{eq52}
\begin{gathered}
  {\mathbf{YX}}_0^H{\left( {{{\mathbf{X}}_0}{{\mathbf{Y}}^H}{\mathbf{YX}}_0^H} \right)^{ - 1}}{{\mathbf{X}}_0}{{\mathbf{Y}}^H} \hfill \\
   = {\left( {{{\left( {{{\mathbf{X}}_0}{{\mathbf{Y}}^H}} \right)}^{ - 1}}{{\mathbf{X}}_0}{{\mathbf{Y}}^H}{\mathbf{YX}}_0^H{{\left( {{\mathbf{YX}}_0^H} \right)}^{ - 1}}} \right)^{ - 1}} = {{\mathbf{I}}_N}. \hfill \\
\end{gathered}
\end{equation}
It follows that
\begin{equation}\label{eq53}
\begin{gathered}
  {T_{Rs}}\left({{\mathbf{Y}}}\right)  = \frac{2}{{{N_0}}}\operatorname{tr} \left( {\left( {{{\mathbf{I}}_L} - \frac{M}{{L{P_R}}}{\mathbf{X}}_0^H{{\mathbf{X}}_0}} \right){{\mathbf{Y}}^H}{\mathbf{Y}}} \right) \hfill \\
   = \frac{2}{{{N_0}}}\operatorname{tr} \left( {{\mathbf{Y}}\left( {{{\mathbf{I}}_L} - \frac{M}{{L{P_R}}}{\mathbf{X}}_0^H{{\mathbf{X}}_0}} \right){{\mathbf{Y}}^H}} \right) \hfill \\
   \triangleq \frac{2}{{{N_0}}}\operatorname{tr} \left( {{\mathbf{YP}}{{\mathbf{Y}}^H}} \right) \mathop  \gtrless \limits_{{\mathcal{H}_0}}^{{\mathcal{H}_1}} \gamma \hfill \\
\end{gathered}
\end{equation}
is the Rao detector under this special case. It can be seen that (\ref{eq53}) is also a quadratic form in Gaussian variables. Interestingly, the matrix ${\mathbf{P}} = {{\mathbf{I}}_L} - \frac{M}{{L{P_R}}}{\mathbf{X}}_0^H{{\mathbf{X}}_0}$ is a projection matrix, which projects any vector to the null-space of ${{\mathbf{X}}_0^H}$. Therefore, we have
\begin{equation}\label{eq54}
  \operatorname{tr} \left( {{\mathbf{G}}{{\mathbf{X}}_0}{\mathbf{PX}}_0^H{{\mathbf{G}}^H}} \right) = 0, \hfill \\
\end{equation}
which leads to
\begin{equation}\label{eq55}
\operatorname{tr} \left( {{\mathbf{G}}{{\mathbf{X}}_1}{\mathbf{PX}}_1^H{{\mathbf{G}}^H}} \right) \ge  0 = \operatorname{tr} \left( {{\mathbf{G}}{{\mathbf{X}}_0}{\mathbf{PX}}_0^H{{\mathbf{G}}^H}} \right).
\end{equation}
The above equations (\ref{eq54}) and (\ref{eq55}) can be viewed as the hypothesis testing for the noise-free scenario, where we see that the Rao detector (\ref{eq53}) is effective in differentiating the two hypotheses. By adding the Gaussian noise to ${\mathbf{G}}{{\mathbf{X}}_1}$ and ${\mathbf{G}}{{\mathbf{X}}_0}$, it can be inferred that $T_{Rs}\left({{\mathbf{Y}}};\mathcal{H}_1\right) \ge T_{Rs}\left({{\mathbf{Y}}};\mathcal{H}_0\right)$ with a high probability in the high SNR regime, which makes the detector (\ref{eq53}) valid.
\begin{proposition}
$T_{Rs}$ subjects to central and non-central chi-squared distributions under $\mathcal{H}_0$ and $\mathcal{H}_1$, respectively, which are given as
\begin{equation}\label{eq56}
{T_{Rs}} \sim \left\{ \begin{gathered}
  {\mathcal{H}_0}:\mathcal{X}_{K}^2, \hfill \\
  {\mathcal{H}_1}:\mathcal{X}_{K}^2\left( {{\mu}} \right), \hfill \\
\end{gathered}  \right.
\end{equation}
where ${\mu} = \frac{2}{{{N_0}}}\operatorname{tr} \left( {{\mathbf{G}}{{\mathbf{X}}_1}\left( {{{\mathbf{I}}_L} - \frac{M}{{L{P_R}}}{\mathbf{X}}_0^H{{\mathbf{X}}_0}} \right){\mathbf{X}}_1^H{{\mathbf{G}}^H}} \right)$ is the non-centrality parameter, and $K = 2N\left( {L - M} \right)$ represents the DoFs of the distributions.
\end{proposition}
\begin{proof}
See Appendix C.
\end{proof}
Similar to (\ref{eq50}), the decision error probability at the \emph{i}-th PRI for the special Rao detector (\ref{eq53}) is given by
\begin{equation}\label{eq61}
P_{Rs}^{\left(i\right)} = \left( {1 - {\mathcal{F}_{\mathcal{X}_K^2}}\left( \gamma  \right)} \right)\left( {1 - P_D^{\left(i-1\right)}} \right) + {\mathcal{F}_{\mathcal{X}_K^2\left( \mu  \right)}}\left( \gamma  \right)P_D^{\left(i-1\right)},
\end{equation}
where ${\mathcal{F}_{\mathcal{X}_K^2}}$ and ${\mathcal{F}_{\mathcal{X}_K^2\left( \mu  \right)}}$ are the CDFs of central and non-central chi-squared distributions, respectively.

\subsection{Channel Estimation Performance for NLoS Channels}
As discussed in Sec. IV, there are no closed-form solutions for the estimations of the AoA under the LoS channel. Hence, we only consider the channel estimation performance for the NLoS channel case, where the MSE is used as the performance metric. By denoting the estimated channel as ${\mathbf{\hat G}} = {\mathbf{Y}}{{\mathbf{X}}^H}{\left( {{\mathbf{X}}{{\mathbf{X}}^H}} \right)^{ - 1}}$, the squared error can be given in the form
\begin{equation}\label{eq70}
\begin{gathered}
  \phi  = \left\| {{\mathbf{\hat G}} - {\mathbf{G}}} \right\|_F^2 = \left\| {{\mathbf{Y}}{{\mathbf{X}}^H}{{\left( {{\mathbf{X}}{{\mathbf{X}}^H}} \right)}^{ - 1}} - {\mathbf{G}}} \right\|_F^2 \hfill \\
   = \left\| {{{\left( {{\mathbf{X}}{{\mathbf{X}}^H}} \right)}^{ - 1}}{\mathbf{X}}{{\mathbf{Y}}^H} - {{\mathbf{G}}^H}} \right\|_F^2. \hfill \\
\end{gathered}
\end{equation}
Let us define
\begin{equation}\label{eq71}
\begin{gathered}
  {\mathbf{\bar y}} = \operatorname{vec} \left( {{{\mathbf{Y}}^H}} \right) \sim \mathcal{C}\mathcal{N}\left( {\operatorname{vec} \left( {{{\mathbf{X}}^H}{{\mathbf{G}}^H}} \right),{N_0}{{\mathbf{I}}_{NL}}} \right), \hfill \\
  {\mathbf{T}} = {{\mathbf{I}}_N} \otimes {\left( {{\mathbf{X}}{{\mathbf{X}}^H}} \right)^{ - 1}}{\mathbf{X}},{\mathbf{\bar g}} = \operatorname{vec} \left( {{{\mathbf{G}}^H}} \right). \hfill \\
\end{gathered}
\end{equation}
Then, (\ref{eq70}) can be simplified as
\begin{equation}\label{eq72}
\phi  = {\left\| {{\mathbf{T\bar y}} - {\mathbf{\bar g}}} \right\|^2}.
\end{equation}
Based on basic statistics and linear algebra, we also have
\begin{equation}\label{eq73}
{{\mathbf{y}}_{eq}} \triangleq {\mathbf{T\bar y}} - {\mathbf{\bar g}} \sim \mathcal{C}\mathcal{N}\left( {{\mathbf{0}},{N_0}{\mathbf{T}}{{\mathbf{T}}^H}} \right),
\end{equation}
where
\begin{equation}\label{eq74}
\begin{gathered}
  {\mathbf{T}}{{\mathbf{T}}^H} = {{\mathbf{I}}_N} \otimes {\left( {{\mathbf{X}}{{\mathbf{X}}^H}} \right)^{ - 1}}{\mathbf{X}} \cdot {{\mathbf{I}}_N} \otimes {{\mathbf{X}}^H}{\left( {{\mathbf{X}}{{\mathbf{X}}^H}} \right)^{ - 1}} \hfill \\
   = {{\mathbf{I}}_N} \otimes {\left( {{\mathbf{X}}{{\mathbf{X}}^H}} \right)^{ - 1}}. \hfill \\
\end{gathered}
\end{equation}
Based on the above, the MSE of the channel estimation can be obtained as
\begin{equation}\label{eq75}
\begin{gathered}
  \mathbb{E}\left( \phi  \right) = \mathbb{E}\left( {{{\left\| {{{\mathbf{y}}_{eq}}} \right\|}^2}} \right) = \mathbb{E}\left( {\operatorname{tr} \left( {{{\mathbf{y}}_{eq}}{\mathbf{y}}_{eq}^H} \right)} \right) = \operatorname{tr} \left( {\mathbb{E}\left( {{{\mathbf{y}}_{eq}}{\mathbf{y}}_{eq}^H} \right)} \right) \hfill \\
   = {N_0}\operatorname{tr} \left( {{{\mathbf{I}}_N} \otimes {{\left( {{\mathbf{X}}{{\mathbf{X}}^H}} \right)}^{ - 1}}} \right) = \frac{{{N_0}N}}{L}\operatorname{tr} \left( {{\mathbf{R}}_X^{ - 1}} \right). \hfill \\
\end{gathered}
\end{equation}
In the case that the directional waveform is transmitted, we have ${{\mathbf{R}}_X} = \frac{1}{L}{{\mathbf{X}}_1}{\mathbf{X}}_1^H$. For the orthogonal transmission, the MSE can be given as
\begin{equation}\label{eq76}
\mathbb{E}\left( \phi  \right) = \frac{{{N_0}N}}{L}\operatorname{tr} \left( {{{\left( {\frac{{{P_R}}}{M}{\mathbf{I}}_M} \right)}^{ - 1}}} \right) = \frac{{{N_0}{M^2}N}}{{L{P_R}}}.
\end{equation}
It is clear from (\ref{eq75}) that the MSE is determined by the covariance matrix and the length of the radar waveforms, as well as the antenna number at the BS.

\subsection{Energy Detection for LoS Channels}
In this subsection, we analyze the performance of the energy detector (\ref{eq32}). First of all, let us rewrite (\ref{eq32}) as
\begin{equation}\label{eq62}
\frac{2}{{{N_0}}}\operatorname{tr} \left( {{\mathbf{Y}}{{\mathbf{Y}}^H}} \right) = 2{{\mathbf{\tilde y}}^H}{\mathbf{\tilde y}},
\end{equation}
where ${\mathbf{\tilde y}}\sim\mathcal{C}\mathcal{N}\left( {{\mathbf{d}},{{\mathbf{I}}_{NL}}} \right)$ is given in (\ref{eq42}), with $\mathbf{d}$ being defined as
\begin{equation}\label{eq63}
{\mathbf{d}} = \left\{ \begin{gathered}
  {\mathcal{H}_0}:{{\operatorname{vec} \left( {{\alpha ^*}{\mathbf{X}}_0^H{\mathbf{a}}\left( \theta  \right){{\mathbf{b}}^H}\left( \theta  \right)} \right)} \mathord{\left/
 {\vphantom {{\operatorname{vec} \left( {{\alpha ^*}{\mathbf{X}}_0^H{\mathbf{a}}\left( \theta  \right){{\mathbf{b}}^H}\left( \theta  \right)} \right)} {\sqrt {{N_0}} }}} \right.
 \kern-\nulldelimiterspace} {\sqrt {{N_0}} }}, \hfill \\
  {\mathcal{H}_1}:{{\operatorname{vec} \left( {{\alpha ^*}{\mathbf{X}}_1^H{\mathbf{a}}\left( \theta  \right){{\mathbf{b}}^H}\left( \theta  \right)} \right)} \mathord{\left/
 {\vphantom {{\operatorname{vec} \left( {{\alpha ^*}{\mathbf{X}}_1^H{\mathbf{a}}\left( \theta  \right){{\mathbf{b}}^H}\left( \theta  \right)} \right)} {\sqrt {{N_0}} }}} \right.
 \kern-\nulldelimiterspace} {\sqrt {{N_0}} }}. \hfill \\
\end{gathered}  \right.
\end{equation}
Eq. (\ref{eq62}) is the sum of the squared Gaussian variables, which subjects to the non-central chi-squared distribution \cite{kay1998fundamentals}. Recall the proof of Proposition 3. By replacing the matrix $\mathbf{P}$ in (\ref{eq57}) as the identity matrix ${\mathbf{I}}_L$, we obtain the non-centrality parameters under two hypotheses as
\begin{equation}\label{eq64}
\begin{gathered}
  {\varepsilon _0} = \frac{{2{{\left| \alpha  \right|}^2}}}{{{N_0}}}\operatorname{tr} \left( {{\mathbf{b}}\left( \theta  \right){{\mathbf{a}}^H}\left( \theta  \right){{\mathbf{X}}_0}{\mathbf{X}}_0^H{\mathbf{a}}\left( \theta  \right){{\mathbf{b}}^H}\left( \theta  \right)} \right) \hfill \\
   = \frac{{2{{\left| \alpha  \right|}^2}NL{P_R}}}{{{N_0}}}, \hfill \\
\end{gathered}
\end{equation}
\begin{equation}\label{eq65}
  {\varepsilon _1} = \frac{{2{{\left| \alpha  \right|}^2}}}{{{N_0}}}\operatorname{tr} \left( {{\mathbf{b}}\left( \theta  \right){{\mathbf{a}}^H}\left( \theta  \right){{\mathbf{X}}_1}{\mathbf{X}}_1^H{\mathbf{a}}\left( \theta  \right){{\mathbf{b}}^H}\left( \theta  \right)} \right).
\end{equation}
The DoFs of both distributions are obtained as
\begin{equation}\label{eq66}
\kappa  = 2\operatorname{rank} \left( {{{\mathbf{I}}_{NL}}} \right) = 2NL.
\end{equation}
Given any ${\tilde \eta} \ge {\tilde \gamma} \ge 0$ as the thresholds for the energy detector (\ref{eq32}), it follows that
\begin{equation}\label{eq67}
\frac{1}{L}\operatorname{tr} \left( {{\mathbf{Y}}{{\mathbf{Y}}^H}} \right) \in \left[ {{\tilde \gamma} ,{\tilde \eta} } \right] \Leftrightarrow \frac{2}{{{N_0}}}\operatorname{tr} \left( {{\mathbf{Y}}{{\mathbf{Y}}^H}} \right) \in \left[ {\frac{{2L{\tilde \gamma} }}{{{N_0}}},\frac{{2L{\tilde \eta} }}{{{N_0}}}} \right].
\end{equation}
Let $\displaystyle \gamma  \triangleq \frac{{2L\tilde \gamma }}{{{N_0}}},\eta  \triangleq \frac{{2L\tilde \eta }}{{{N_0}}}$. Under the two hypotheses, the probability that the test statistic does not fall into the decision region can be accordingly given by
\begin{equation}\label{eq68}
\begin{gathered}
  P\left( {{T_E}\left( {\mathbf{Y}} \right) \notin \left[ {\tilde \gamma ,\tilde \eta } \right];{\mathcal{H}_0}} \right) = P\left( {\frac{2}{{{N_0}}}\operatorname{tr} \left( {{\mathbf{Y}}{{\mathbf{Y}}^H}} \right) \notin \left[ {\gamma ,\eta } \right];{\mathcal{H}_0}} \right) \hfill \\
   = 1 - \left( {1 - {\mathcal{F}_{\mathcal{X}_\kappa ^2\left( {{\varepsilon _0}} \right)}}\left( \gamma  \right)} \right){\mathcal{F}_{\mathcal{X}_\kappa ^2\left( {{\varepsilon _0}} \right)}}\left( \eta  \right), \hfill \\
  P\left( {{T_E}\left( {\mathbf{Y}} \right) \in \left[ {\tilde \gamma ,\tilde \eta } \right];{\mathcal{H}_1}} \right) = P\left( {\frac{2}{{{N_0}}}\operatorname{tr} \left( {{\mathbf{Y}}{{\mathbf{Y}}^H}} \right) \in \left[ {\gamma ,\eta } \right];{\mathcal{H}_1}} \right) \hfill \\
   = \left( {1 - {\mathcal{F}_{\mathcal{X}_\kappa ^2\left( {{\varepsilon _1}} \right)}}\left( \gamma  \right)} \right){\mathcal{F}_{\mathcal{X}_\kappa ^2\left( {{\varepsilon _1}} \right)}}\left( \eta  \right). \hfill \\
\end{gathered}
\end{equation}
Finally, at the \emph{i}-th PRI, the decision error probability for the energy detector is
\begin{equation}\label{eq69}
\begin{gathered}
  P_E^{\left(i\right)} = \left[ {1 - \left( {1 - {\mathcal{F}_{\mathcal{X}_\kappa ^2\left( {{\varepsilon _0}} \right)}}\left( \gamma  \right)} \right){\mathcal{F}_{\mathcal{X}_\kappa ^2\left( {{\varepsilon _0}} \right)}}\left( \eta  \right)} \right]\left( {1 - P_D^{\left(i-1\right)}} \right) \hfill \\
  \;\;\;\;\;\;\;\;\;\;\;\;\;\;\;\;\;\;\; + \left( {1 - {\mathcal{F}_{\mathcal{X}_\kappa ^2\left( {{\varepsilon _1}} \right)}}\left( \gamma  \right)} \right){\mathcal{F}_{\mathcal{X}_\kappa ^2\left( {{\varepsilon _1}} \right)}}\left( \eta  \right)P_D^{\left(i-1\right)}. \hfill \\
\end{gathered}
\end{equation}

\subsection{Discussion on the Hypothesis Testing Thresholds}
It is worth highlighting that the performance of all the detectors above relies on the given thresholds. Typically, the threshold is chosen to optimize certain performance metrics, i.e., the decision error probability in our case. Note that the GLRT detector is equivalent to the maximum likelihood ratio. Hence the optimal threshold can be straightforwardly given as $\displaystyle \gamma = \ln \left( {\frac{{1 - P_D^{\left(i-1\right)}}}{{P_D^{\left(i-1\right)}}}} \right)$. Nevertheless, as the true value of ${P_D^{\left(i-1\right)}}$ is unknown to the BS, only the suboptimal threshold $\gamma = 0$ can be adopted.
\\\indent For the Rao and energy detectors, the BS is unable to determine the optimal hypothesis testing thresholds, since it does not know the tracking waveform $\mathbf{X}_1$ under such scenarios. Therefore, the hypothesis testing thresholds can only be obtained by numerical simulations. We address this issue in the next section.

\section{Numerical Results}
In this section, numerical results are provided to verify the effectiveness of the proposed approaches. Below we introduce the parameters used in our simulations.
\begin{enumerate}
\item {\textbf {Radar Waveforms:}} We use $ {{\mathbf{X}}_0} = \sqrt {\frac{{L{P_R}}}{M}} {\mathbf{U}}$ as the radar searching waveform, where ${\mathbf{U}}\in \mathbb{C}^{M\times L}$ is an arbitrarily given unitary matrix. For the tracking waveform ${\mathbf{X}}_1$, we firstly solve the classic 3dB beampattern design problem to obtain the waveform covariance matrix ${\mathbf{R}}\in \mathbb{C}^{M\times M}$, which is \cite{4350230}
    \begin{equation}\label{eq_opt}
    \begin{gathered}
    \mathop {\min }\limits_{t,{\mathbf{R}}} {\kern 1pt} {\kern 1pt} {\kern 1pt}  - t \hfill \\
    s.t.\;\;\;{{\mathbf{a}}^H}\left( {{\theta _0}} \right){\mathbf{Ra}}\left( {{\theta _0}} \right) - {{\mathbf{a}}^H}\left( {{\theta _m}} \right){\mathbf{Ra}}\left( {{\theta _m}} \right) \ge t,\forall {\theta _m} \in \Psi,  \hfill \\
    \;\;\;\;\;\;\;{\kern 1pt} {\kern 1pt} {{\mathbf{a}}^H}\left( {{\theta _1}} \right){\mathbf{Ra}}\left( {{\theta _1}} \right) = {{\mathbf{a}}^H}\left( {{\theta _0}} \right){\mathbf{Ra}}\left( {{\theta _0}} \right)/2, \hfill \\
    \;\;\;\;\;\;\;\;{{\mathbf{a}}^H}\left( {{\theta _2}} \right){\mathbf{Ra}}\left( {{\theta _2}} \right) = {{\mathbf{a}}^H}\left( {{\theta _0}} \right){\mathbf{Ra}}\left( {{\theta _0}} \right)/2,\;\; \hfill \\
    \;\;\;\;\;\;\;\;{\mathbf{R}} \succeq 0,{\mathbf{R}} = {{\mathbf{R}}^H}, \hfill \\
    \;\;\;\;\;\;\;\;\operatorname{diag}\left( {\mathbf{R}} \right) = \frac{{{P_R}{\mathbf{1}}}}{M}, \hfill \\
    \end{gathered}
    \end{equation}
    where $\theta_0$ denotes the azimuth angle of the target, i.e., the location of the radar's mainlobe, whose 3dB beamwidth is determined by $\left(\theta_2-\theta_1\right)$, and $\Psi$ stands for the sidelobe region. According to \cite{4350230}, problem (\ref{eq_opt}) is convex, and can be easily solved via numerical tools. The tracking beampattern generated by (\ref{eq_opt}) can accurately achieve the desired 3dB beamwidth, while maintaining the minimum sidelobe level. We then obtain the tracking waveform $ {{\mathbf{X}}_1}$ by the Cholesky decomposition of $ {{\mathbf{R}}}$. Without loss of generality, we assume that the mainlobe focuses on the angle of $0^{\circ}$, and the desired 3dB beamwidth is $10^{\circ}$.
\item {\textbf {Threshold Setting:}} For the GLRT detectors, we consider both the optimal threshold $\gamma = \ln \left( {\frac{{1 - P_D^{\left(i-1\right)}}}{{P_D^{\left(i-1\right)}}}} \right)$ and its suboptimal counterpart $\gamma = 0$. Since the optimal threshold for Rao test is difficult to obtain, we provide the ergodic empirical thresholds, which are computed by Monte Carlo simulations with a large number of channel realizations, and can guarantee that the \emph{average error probability} is minimized. Meanwhile, we also compute the optimal threshold that corresponds to one single channel realization for $M = N$, where the theoretical error probability is given in (\ref{eq61}). Note that such a threshold is not obtainable in practical scenarios, as it requires the BS to know the channel a priori. In our simulations, it serves as the performance benchmark for the Rao test. For the energy detector under the LoS channel, the empirical thresholds are simply given as $\gamma  = N\left( {\frac{P_R}{2} + {N_0}} \right),\eta  = N\left( {2{P_R} + {N_0}} \right)$, while the performance of the optimal thresholds for one single channel realization is also presented for comparison.
\item {\textbf {Other Parameters:}} For simplicity, we assume that the detection probability of radar is the same at each PRI, namely $P_D^i = P_D, \forall i$. Without loss of generality, we set $P_R = 1$, and define the transmit SNR of radar as $\operatorname{SNR} = P_R/N_0$. Unless otherwise specified, we fix $L = 20$, and assume half-wavelength separation between adjacent antennas.
\end{enumerate}
\subsection{NLoS Channel Scenario}
In this subsection, we assume a Rayleigh fading channel $\mathbf{G}$, i.e., the entries of $\mathbf{G}$ are independent and identically distributed (i.i.d.) and subject to the standard complex Gaussian distribution. We firstly consider the case that $M = N = 16, L = 20, P_D = 0.9$. To understand the impact of the ergodic HT thresholds on the performance of the Rao test, Fig. 4 shows the decision error probability computed through Monte Carlo simulations for increasing values of the HT thresholds. It can be observed that, for each SNR value, the error probability curve has a unique minimum point, which determines the ergodic threshold for the detector. We then use these results for the following Rao test simulations.
\begin{figure}[!htp]
 \centering
 \includegraphics[width=\columnwidth]{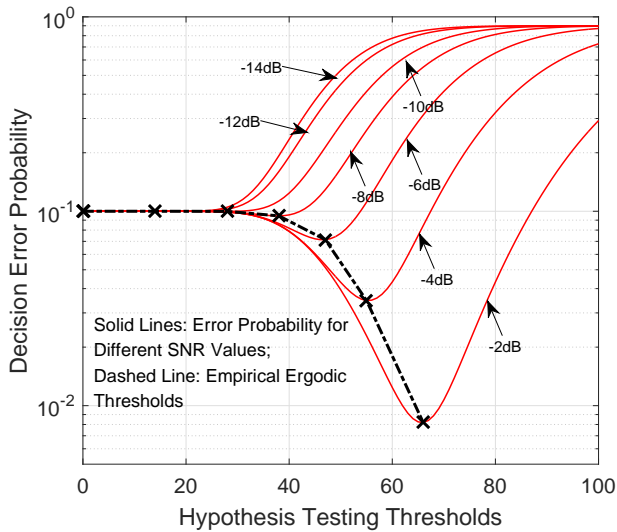}
 \caption{Decision error probability of the Rao test for varying HT thresholds $\gamma$. $M = N = 16, L = 20, P_D = 0.9$.}
 \label{fig:FIG24}
\end{figure}
\\\indent In Fig. 5, the performances of the GLRT and the Rao test are compared under the same parameter configuration of Fig. 4, where the theoretical and simulated curves are denoted by solid and dashed lines, respectively. For GLRT, we employ both the optimal and suboptimal thresholds mentioned above. For the Rao test, we investigate not only the empirical thresholds shown in Fig. 4, but also the optimal thresholds for the specific instantaneous channel realization. It can be noted that the theoretical curves match well with their simulated counterparts for both detectors, which validates our performance analysis of (\ref{eq50}) and (\ref{eq61}) in Sec. V. Moreover, the Rao detector outperforms the GLRT in the low SNR regime, where the associated error probability is close to 0.1. The reason for this is explained as follows. In light of Fig. 4, the optimal threshold for Rao test is close to 0 when the SNR is low. Due to the non-negativity of the Rao test statistic (\ref{eq53}), hypothesis $\mathcal{H}_1$ will always be chosen by the detector, which has the prior probability of $P\left( {{\mathcal{H}_1}} \right) = {P_D} = 0.9$, leading to an error probability of 0.1. It can be further noted that the GLRT statistic (\ref{eq41}) can be either positive or negative. When the SNR is low, the GLRT detector choose randomly from the two hypotheses, resulting in an error probability of 0.5. At the high SNR regime, however, GLRT outperforms the Rao detector, as it employs the information of both $\mathbf{X}_0$ and $\mathbf{X}_1$.
\begin{figure}[!htp]
 \centering
 \includegraphics[width=\columnwidth]{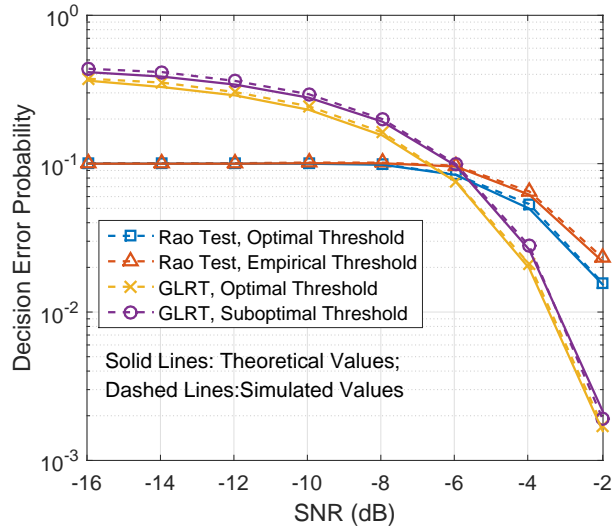}
 \caption{Decision error probability vs. SNR for the GLRT and Rao tests. $M = N = 16, L = 20, P_D = 0.9$.}
 \label{fig:FIG25}
\end{figure}
\\\indent We further show in Fig. 6 the detection performance for $P_D = 0.5$, where we fix $N = 16$, and set $M = 10$ and $M = 16$ respectively. Note that the optimal and the suboptimal thresholds for GLRT are exactly the same, given the prior probability of $0.5$ for each hypothesis. For the Rao test, since the analytical performance for the nonequal-antenna case is intractable, we only show the performance with empirical threshold for $M = 16$. It can be observed that, when $M = 10$, the performance for both detectors are superior to that of the case of $M =16$, which is sensible given that the BS exploits more DoFs for hypothesis testing in the former case. In addition, the GLRT outperforms the Rao test for both low and high SNR regimes. This is because the priori probability for $\mathcal{H}_1$ is now 0.5, leading to an error probability of 0.5 for Rao test for the low SNR regime, which further verifies the correctness of our observations in the analysis of Fig. 5.
\begin{figure}[!htp]
 \centering
 \includegraphics[width=\columnwidth]{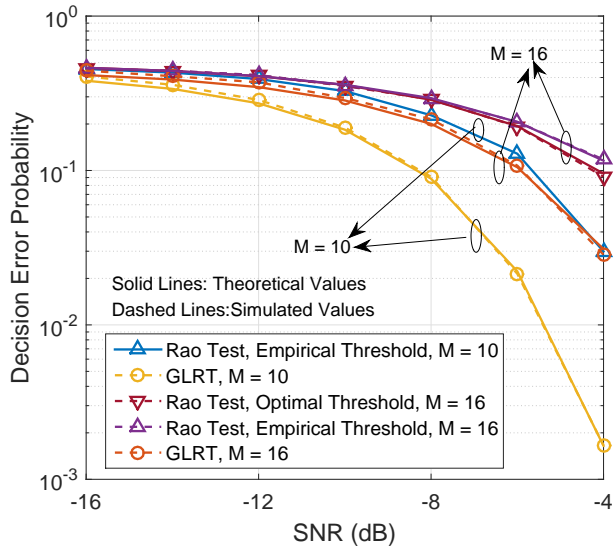}
 \caption{Decision error probability vs. SNR for the GLRT and Rao tests. $M = \left\{10, 16\right\}, N = 16, L = 20, P_D = 0.5$.}
 \label{fig:FIG26}
\end{figure}
\\\indent We investigate the channel estimation performance in Fig. 7, where we fix the radar antenna number as $M = 5$, and increase the BS antennas from $N = 4$ to $N = 20$. Note that the hypothesis testing exploits the power of all the entries in the received signal matrix to make the binary decision, which does not require a high SNR per entry to guarantee a successful outcome. This is very similar to the concept of diversity gain. Nevertheless, for the NLoS channel estimation, we need to estimate each entry individually, where the diversity gain does not exist. For this reason, we fix the SNR at 15dB to achieve the normal estimation performance. It can be seen from Fig. 7 that the theoretical curves match well with the simulated ones, which proves the correctness of (\ref{eq75}) and (\ref{eq76}). Secondly, the MSE increases with the rise of the BS antenna number, owing to the increasing number of the matrix entries to be estimated. Finally, it is worth highlighting that better estimation performance can be achieved by use of the searching waveform $\mathbf{X}_0$ rather than the tracking waveform $\mathbf{X}_1$. This is because the optimal pilot signals are orthogonal waveforms such as $\mathbf{X}_0$ according to the channel estimation theory \cite{1200150}.
\begin{figure}[!htp]
 \centering
 \includegraphics[width=\columnwidth]{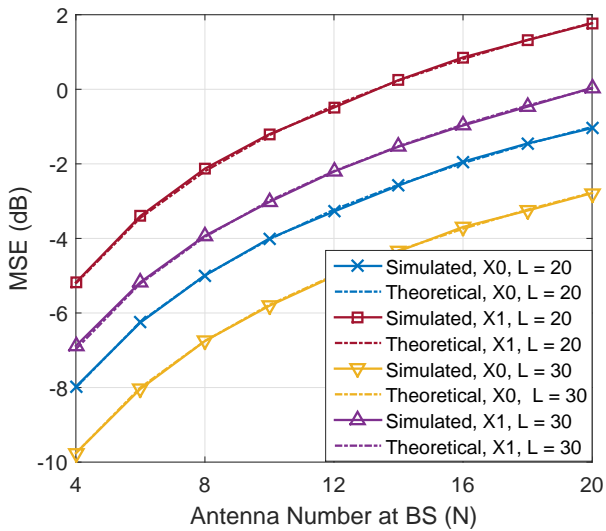}
 \caption{Channel estimation MSE vs. number of antennas at the BS. $M = 5, \text{SNR} = 15\text{dB}$.}
 \label{fig:FIG27}
\end{figure}
\subsection{LoS Channel Scenario}
In this subsection, we show the numerical results for the LoS channel scenario. Unless otherwise specified, we assume that the BS is located at $\theta = 20^{\circ}$ relative to the radar. In each Monte Carlo simulation, a unit-modulus path-loss factor $\alpha$ is randomly generated.
\\\indent We first look at the detection performance of GLRT and ED in Fig. 8 with $M = N = 16, L = 20, P_D = 0.9$. For simplicity, we use ``ED" to refer to the energy detection in Fig. 8. Again, we observe that the theoretical curves match well with their simulated counterparts. It is interesting to see that the energy detector outperforms the GLRT detector under high SNR regime. This is a counter-intuitive behavior, as the GLRT exploits both $\mathbf{X}_0$ and $\mathbf{X}_1$ while the energy detector requires nothing from the radar. However, this result can be explained by realizing that the performance of GLRT is highly dependent on the information contained in the received signals. Specifically, since the LoS channel projects the received signal matrix onto a rank-1 subspace, this breaks down the structure of the transmitted waveforms. In contrast, the energy detection exploits the difference between the two beampatterns, which is equivalent to utilizing the intrinsic structure of $\mathbf{X}_0$ and $\mathbf{X}_1$, and hence leads to better performance.
\begin{figure}[!htp]
 \centering
 \includegraphics[width=\columnwidth]{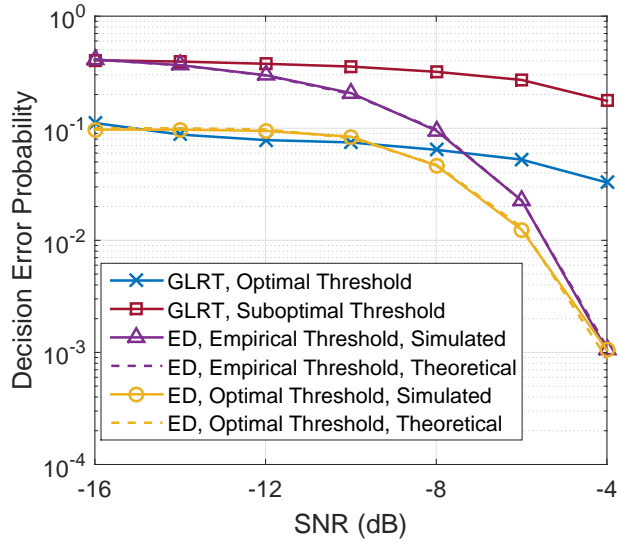}
 \caption{Decision error probability vs. SNR for the GLRT and energy detection HT under a LoS channel. $M = N = 16, L = 20, P_D = 0.9$.}
 \label{fig:FIG28}
\end{figure}
\begin{figure}[!htp]
 \centering
 \includegraphics[width=\columnwidth]{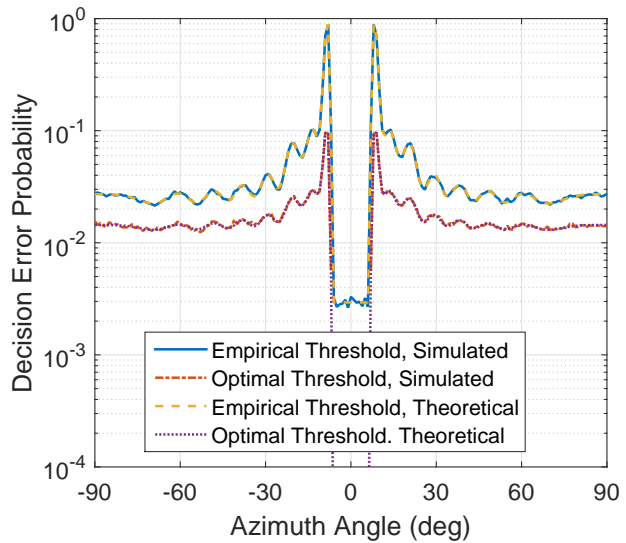}
 \caption{Decision error probability vs. azimuth angle for the energy detection HT under a LoS channel. $M = N = 16, L = 20, P_D = 0.9, \text{SNR} = -6\text{dB}$.}
 \label{fig:FIG29}
\end{figure}
\\\indent As discussed above, the ED exploits the difference between the omnidirectional and directive beampatterns, in which case the performance of the energy detector relies on the angle of the BS relative to the radar. We therefore show in Fig. 9 the decision error probability at the BS by varying its azimuth angle $\theta$, where the SNR is set as $-6\text{dB}$, and the detectors with both optimal and empirical thresholds are considered. Interestingly, all of the curves in the figure show a shape similar to that of the tracking beampattern in Fig. 3. This is because the detection performance of ED is mainly determined by the power gap between the two beampatterns. In the mainlobe area, we see that the error performance is better than that of the other areas, owing to the largest power gap within omnidirectional and directive antenna patterns in this region, as illustrated in Fig. 3. Finally, as predicted in Sec. IV-B, the detection performance becomes worse if the BS is located at an angle that falls into the ambiguity region, where the two beampatterns are unable to be effectively differentiated. Fortunately, the BS is unlikely to be located in such area since the it only occupies a small portion of the whole space.
\begin{figure}[!htp]
 \centering
 \includegraphics[width=\columnwidth]{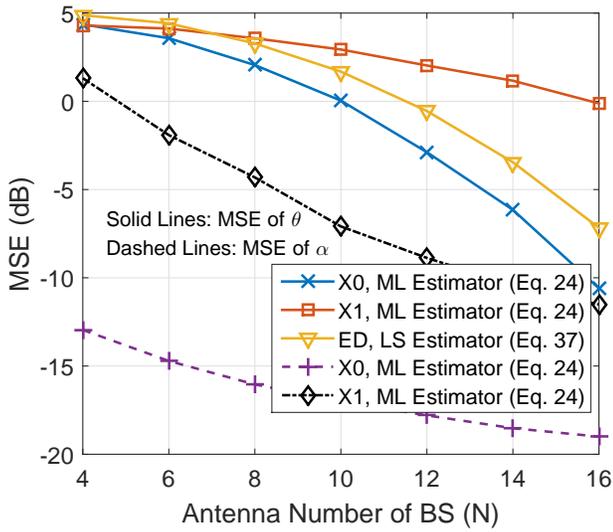}
 \caption{Channel estimation MSE vs. number of antennas at the BS for LoS scenario, $M = 4, L = 20, \text{SNR} = -6\text{dB}$.}
 \label{fig:FIG210}
\end{figure}
\\\indent Fig. 10 shows the channel estimation performance for the LoS scenario with an increasing number of BS antennas, where $M = 4, L = 20, \text{SNR} = -6\text{dB}$. In this figure, the maximum likelihood (ML) and the least-squares (LS) estimators (\ref{eq24}) and (\ref{eq37}) are employed for the cases of known and unknown waveforms, respectively. In contrast to the NLoS channel shown in Fig. 7, Fig. 10 illustrates that the MSE of both the estimated $\theta$ and $\alpha$ decreases with the increase of the BS antennas under the LoS channel. This is because $\theta$ and $\alpha$ are the only two parameters to be estimated, which can be more accurately obtained by increasing the DoFs at the BS. It can be again observed that the accuracy of $\mathbf{X}_0$ is superior to that of $\mathbf{X}_1$ when the ML estimator is used, thanks to the orthogonal nature of the searching waveform. Nevertheless, we still need to identify the working mode of the radar before we can estimate the channel parameters. Moreover, there exists a 3dB performance gap between the LS estimator and the ML estimator using $\mathbf{X}_0$. This is because the LS estimator (\ref{eq24}) is solely based on the searching waveform $\mathbf{X}_0$, which is definitely worse than the associated ML estimator, as the latter is typically the optimal estimator in a statistical sense. Even so, the performance of the LS estimator is satisfactory enough, as it does not require any information of the radar waveforms.

\section{Conclusions}
This paper deals with the issue of interfering channel estimation for radar and cellular coexistence, where we assume that the radar switches randomly between the searching and tracking modes, and the BS is attempting to estimate the radar-cellular interfering channel by use of the radar probing waveforms. To acquire the channel state information, the BS firstly identifies the working mode of the radar by use of hypothesis testing approaches, and then estimates the channel parameters. For completeness, both the LoS and NLoS channels are considered, where different detectors are proposed as per the available priori knowledge at the BS, namely GLRT, Rao test and energy detection. As a step further, the theoretical performance of the proposed approaches are analyzed in detail using statistical techniques. Our simulations show that the theoretical curves match well with the numerical results, and that the BS can effectively estimate the interfering channel, even with limited information from the radar.


%

 \appendices
\section{Proof of Proposition 1}
In the LoS channel case, the logarithmic probability density function (log-PDF) of the received signal matrix can be given as
\begin{equation}\label{eq77}
\begin{small}
\begin{gathered}
  \ln p\left({\mathbf{Y}}\right) =  - NL\ln \pi {N_0} \hfill \\
   - \frac{1}{{{N_0}}}\operatorname{tr} \left( {{{\left( {{\mathbf{Y}} - \alpha {\mathbf{b}}\left( \theta  \right){{\mathbf{a}}^H}\left( \theta  \right){\mathbf{X}}} \right)}^H}\left( {{\mathbf{Y}} - \alpha {\mathbf{b}}\left( \theta  \right){{\mathbf{a}}^H}\left( \theta  \right){\mathbf{X}}} \right)} \right).
\end{gathered}
\end{small}
\end{equation}
According to \cite{kay1998fundamentals}, the FIM can be partitioned as
\begin{equation}\label{eq78}
{\mathbf{J}}\left( {\mathbf{\Theta }} \right) = \left[ {\begin{array}{*{20}{c}}
  {{{\mathbf{J}}_{rr}}}&{{{\mathbf{J}}_{rs}}} \\
  {{{\mathbf{J}}_{sr}}}&{{{\mathbf{J}}_{ss}}}
\end{array}} \right],
\end{equation}
where
\begin{equation}\label{eq79}
\begin{gathered}
  {{\mathbf{J}}_{rr}} = \mathbb{E}\left( {\frac{{\partial \ln p}}{{\partial {{\operatorname{vec} }^{\text{*}}}\left( {\mathbf{X}} \right)}}\frac{{\partial \ln p}}{{\partial {{\operatorname{vec} }^T}\left( {\mathbf{X}} \right)}}} \right) \hfill \\
   = \frac{{4N{{\left| \alpha  \right|}^2}}}{{{N_0}}}{{\mathbf{I}}_L} \otimes {{\mathbf{a}}^*}\left( \theta  \right){{\mathbf{a}}^T}\left( \theta  \right) \in {\mathbb{C}^{ML \times ML}}. \hfill \\
\end{gathered}
\end{equation}
Let ${{\bm{\theta }}_s} = {\left[ {\alpha ,\theta } \right]^T} \in {\mathbb{C}^{2 \times 1}}$ be the nuisance parameters, then
\begin{equation}\label{eq80}
\begin{gathered}
  {{\mathbf{J}}_{rs}} = \mathbb{E}\left( {\frac{{\partial \ln p}}{{\partial {{\operatorname{vec} }^{\text{*}}}\left( {\mathbf{X}} \right)}}{{\left( {\frac{{\partial \ln p}}{{\partial {{\bm{\theta }}_s}}}} \right)}^T}} \right) \in {\mathbb{C}^{ML \times 2}}, \hfill \\
  {{\mathbf{J}}_{sr}} = {\mathbf{J}}_{rs}^H \in {\mathbb{C}^{2 \times ML}}, \hfill \\
  {{\mathbf{J}}_{ss}} = \mathbb{E}\left( {\frac{{\partial \ln p}}{{\partial {\bm{\theta }}_s^*}}{{\left( {\frac{{\partial \ln p}}{{\partial {{\bm{\theta }}_s}}}} \right)}^T}} \right) \in {\mathbb{C}^{2 \times 2}}. \hfill \\
\end{gathered}
\end{equation}
From (\ref{eq80}) and (\ref{eq81}), it can be observed that
\begin{equation}\label{eq81}
\begin{gathered}
  \operatorname{rank} \left( {{{\mathbf{J}}_{rr}}} \right) = L, \hfill \\
  \operatorname{rank} \left( {{{\mathbf{J}}_{rs}}} \right) \le 2, \operatorname{rank} \left( {{{\mathbf{J}}_{sr}}} \right) \le 2, \operatorname{rank} \left( {{{\mathbf{J}}_{ss}}} \right) \le 2. \hfill \\
\end{gathered}
\end{equation}
To compute the upper-left partition of the inverse FIM, let us define
\begin{equation}\label{eq82}
{\mathbf{\bar J}} = {{\mathbf{J}}_{rr}}\left( {{\mathbf{\tilde \Theta }}} \right) - {{\mathbf{J}}_{rs}}\left( {{\mathbf{\tilde \Theta }}} \right){\mathbf{J}}_{ss}^{ - 1}\left( {{\mathbf{\tilde \Theta }}} \right){{\mathbf{J}}_{sr}}\left( {{\mathbf{\tilde \Theta }}} \right).
\end{equation}
By using the property of the rank operator, and recalling that $L \ge M >2$, we have
\begin{equation}\label{eq83}
\operatorname{rank} \left( {\mathbf{\bar J}} \right) \le L + 2 < ML,
\end{equation}
which indicates that ${\mathbf{\bar J}}\in \mathbb{C}^{ML \times ML}$ is a singular matrix and is thus non-invertible. Hence, the Rao test statistic does not exist. This completes the proof.

\section{Proof of Proposition 2}
In the NLoS channel case, the log-PDF can be given as
 \begin{equation}\label{eq84}
\ln p =  - NL\ln \pi {N_0} - \frac{1}{{{N_0}}}\operatorname{tr} \left( {{{\left( {{\mathbf{Y}} - {\mathbf{GX}}} \right)}^H}\left( {{\mathbf{Y}} - {\mathbf{GX}}} \right)} \right).
 \end{equation}
 To compute the Fisher Information, we calculate the derivatives as
\begin{equation}\label{eq85}
\begin{small}
\begin{gathered}
  \frac{{\partial \ln p}}{{\partial \operatorname{vec} \left( {\mathbf{X}} \right)}} = \frac{2}{{{N_0}}}\left( {{{\mathbf{I}}_L} \otimes {{\mathbf{G}}^H}} \right){\mathbf{z}},\frac{{\partial \ln p}}{{\partial {{\operatorname{vec} }^{\text{*}}}\left( {\mathbf{X}} \right)}} = \frac{2}{{{N_0}}}\left( {{{\mathbf{I}}_L} \otimes {{\mathbf{G}}^T}} \right){{\mathbf{z}}^*}, \hfill \\
  \frac{{\partial \ln p}}{{\partial \operatorname{vec} \left( {\mathbf{G}} \right)}} = \frac{2}{{{N_0}}}\left( {{{\mathbf{X}}^*} \otimes {{\mathbf{I}}_N}} \right){\mathbf{z}},\frac{{\partial \ln p}}{{\partial {{\operatorname{vec} }^*}\left( {\mathbf{G}} \right)}} = \frac{2}{{{N_0}}}\left( {{\mathbf{X}} \otimes {{\mathbf{I}}_N}} \right){{\mathbf{z}}^*}, \hfill \\
\end{gathered}
\end{small}
\end{equation}
where ${\mathbf{z}} = \operatorname{vec} \left( {{\mathbf{Y}} - {\mathbf{GX}}} \right)$.  Recalling (\ref{eq78})-(\ref{eq80}), and by using the fact that $\mathbb{E}\left( {{{\mathbf{z}}^*}{{\mathbf{z}}^T}} \right) = {N_0}{{\mathbf{I}}_{NL}}$, we have
\begin{equation}\label{eq86}
\begin{gathered}
  {{\mathbf{J}}_{rr}} = \mathbb{E}\left( {\frac{{\partial \ln p}}{{\partial {{\operatorname{vec} }^{\text{*}}}\left( {\mathbf{X}} \right)}}\frac{{\partial \ln p}}{{\partial {{\operatorname{vec} }^T}\left( {\mathbf{X}} \right)}}} \right) \hfill \\
   = \frac{4}{{{N_0^2}}}\left( {{{\mathbf{I}}_L} \otimes {{\mathbf{G}}^T}} \right)\mathbb{E}\left( {{{\mathbf{z}}^*}{{\mathbf{z}}^T}} \right)\left( {{{\mathbf{I}}_L} \otimes {{\mathbf{G}}^*}} \right) \hfill \\
   = \frac{4}{{{N_0}}}{{\mathbf{I}}_L} \otimes {{\mathbf{G}}^T}{{\mathbf{G}}^*}, \hfill \\
\end{gathered}
\end{equation}
\begin{equation}\label{eq87}
{{\mathbf{J}}_{rs}} = \mathbb{E}\left( {\frac{{\partial \ln p}}{{\partial {{\operatorname{vec} }^{\text{*}}}\left( {\mathbf{X}} \right)}}\frac{{\partial \ln p}}{{\partial {{\operatorname{vec} }^T}\left( {\mathbf{G}} \right)}}} \right) = \frac{4}{{{N_0}}}{{\mathbf{X}}^H} \otimes {{\mathbf{G}}^T},
\end{equation}
\begin{equation}\label{eq88}
{{\mathbf{J}}_{sr}} = {\mathbf{J}}_{rs}^H = \frac{4}{{{N_0}}}{\mathbf{X}} \otimes {{\mathbf{G}}^*},
\end{equation}
\begin{equation}\label{eq89}
{{\mathbf{J}}_{ss}} = \mathbb{E}\left( {\frac{{\partial \ln p}}{{\partial {{\operatorname{vec} }^{\text{*}}}\left( {\mathbf{G}} \right)}}\frac{{\partial \ln p}}{{\partial {{\operatorname{vec} }^T}\left( {\mathbf{G}} \right)}}} \right) = \frac{4}{{{N_0}}}{\mathbf{X}}{{\mathbf{X}}^H} \otimes {{\mathbf{I}}_N}.
\end{equation}
The FIM can be therefore expressed as
\begin{equation}\label{eq90}
{\mathbf{J}}\left( {\mathbf{\Theta }} \right) = \frac{4}{{{N_0}}}\left[ {\begin{array}{*{20}{c}}
  {{{\mathbf{I}}_L} \otimes {{\mathbf{G}}^T}{{\mathbf{G}}^*}}&{{{\mathbf{X}}^H} \otimes {{\mathbf{G}}^T}} \\
  {{\mathbf{X}} \otimes {{\mathbf{G}}^*}}&{{\mathbf{X}}{{\mathbf{X}}^H} \otimes {{\mathbf{I}}_N}}
\end{array}} \right].
\end{equation}
By recalling the definition of ${{\mathbf{\tilde \Theta }}}$, and noting that ${{\mathbf{X}}_0}{\mathbf{X}}_0^H = \frac{{L{P_R}}}{M}{{\mathbf{I}}_M} \triangleq \rho {{\mathbf{I}}_M}$, we have
\begin{equation}\label{eq91}
\begin{small}
\begin{gathered}
  {\left[ {{{\mathbf{J}}^{ - 1}}\left( {{\mathbf{\tilde \Theta }}} \right)} \right]_{{{\mathbf{\theta }}_r}{{\mathbf{\theta }}_r}}}\hfill \\
   = {\left( {{{\mathbf{J}}_{rr}}\left( {{\mathbf{\tilde \Theta }}} \right) - {{\mathbf{J}}_{rs}}\left( {{\mathbf{\tilde \Theta }}} \right){\mathbf{J}}_{ss}^{ - 1}\left( {{\mathbf{\tilde \Theta }}} \right){{\mathbf{J}}_{sr}}\left( {{\mathbf{\tilde \Theta }}} \right)} \right)^{ - 1}} \hfill \\
   = \frac{{{N_0}}}{4}{\left( {{{\mathbf{I}}_L} \otimes {\mathbf{\hat G}}_0^T{\mathbf{\hat G}}_0^* - \frac{1}{\rho }\left( {{\mathbf{X}}_0^H \otimes {\mathbf{\hat G}}_0^T} \right){{\mathbf{I}}_{MN}}\left( {{{\mathbf{X}}_0} \otimes {\mathbf{\hat G}}_0^*} \right)} \right)^{ - 1}} \hfill \\
   = \frac{{{N_0}}}{4}{\left( {\left( {{{\mathbf{I}}_L} - \frac{1}{\rho }{\mathbf{X}}_0^H{{\mathbf{X}}_0}} \right) \otimes \left( {{\mathbf{\hat G}}_0^T{\mathbf{\hat G}}_0^*} \right)} \right)^{ - 1}}, \hfill \\
\end{gathered}
\end{small}
\end{equation}
where $\rho  = \frac{{L{P_R}}}{M}$, and ${\mathbf{\hat G}}_0$ is given by (\ref{eq13}). By using (\ref{eq13}), (\ref{eq18}), (\ref{eq85}), and (\ref{eq91}), the Rao test statistic can be expressed as (\ref{eq51}), which completes the proof.

\section{Proof of Proposition 3}
We first rewrite (\ref{eq53}) as
\begin{equation}\label{eq57}
{T_{Rs}}\left({{\mathbf{Y}}}\right) = \frac{2}{{{N_0}}}\operatorname{tr} \left( {{\mathbf{YP}}{{\mathbf{Y}}^H}} \right) = 2{{{\mathbf{\tilde y}}}^H}\left( {{{\mathbf{I}}_N} \otimes {\mathbf{P}}} \right){\mathbf{\tilde y}},
\end{equation}
where ${\mathbf{\tilde y}}$ is defined in (\ref{eq42}). In this expression, both the real and imaginary parts of $\sqrt{2}{\mathbf{\tilde y}}$ subject to the standard normal distribution. Since ${{{\mathbf{I}}_N} \otimes {\mathbf{P}}}$ is also an idempotent matrix, (\ref{eq57}) subjects to non-central chi-squared distribution under both hypotheses \cite{kay1998fundamentals}. Under $\mathcal{H}_0$, the non-centrality parameter is given by
\begin{equation}\label{eq58}
{\mu _0} = \frac{2}{{{N_0}}}\operatorname{tr} \left( {{\mathbf{G}}{{\mathbf{X}}_0}\left( {{{\mathbf{I}}_L} - \frac{M}{{L{P_R}}}{\mathbf{X}}_0^H{{\mathbf{X}}_0}} \right){\mathbf{X}}_0^H{{\mathbf{G}}^H}} \right) = 0,
\end{equation}
which indicates that $T_{Rs}\left({{\mathbf{Y}}};\mathcal{H}_0\right)$ is in fact central chi-squared distributed. Under $\mathcal{H}_1$, the non-centrality parameter is given as
\begin{equation}\label{eq59}
{\mu} = \frac{2}{{{N_0}}}\operatorname{tr} \left( {{\mathbf{G}}{{\mathbf{X}}_1}\left( {{{\mathbf{I}}_L} - \frac{M}{{L{P_R}}}{\mathbf{X}}_0^H{{\mathbf{X}}_0}} \right){\mathbf{X}}_1^H{{\mathbf{G}}^H}} \right).
\end{equation}
The DoFs of the two distributions are given by
\begin{equation}\label{eq60}
\begin{gathered}
  K  = 2\operatorname{rank} \left( {{{\mathbf{I}}_N} \otimes {\mathbf{P}}} \right) \hfill \\
   = 2N\operatorname{rank} \left( {\mathbf{P}} \right) = 2N\operatorname{tr} \left( {\mathbf{P}} \right) = 2N\left( {L - M} \right),\hfill \\
\end{gathered}
\end{equation}
where we use the property of the idempotent matrix that $\operatorname{rank}\left( {\mathbf{P}} \right) = \operatorname{tr} \left( {\mathbf{P}} \right)$. This completes the proof.



\ifCLASSOPTIONcaptionsoff
  \newpage
\fi



\bibliographystyle{IEEEtran}
\bibliography{IEEEabrv,RadCom_Ref}

\end{document}